\documentclass[a4paper,onecolumn,11pt,accepted=2023-01-21]{quantumarticle}
\pdfoutput=1

\usepackage{amsmath,amsfonts,amssymb,amsthm,braket}
\allowdisplaybreaks
\usepackage[margin=0.63in]{geometry}
\usepackage[colorlinks,
            linkcolor = blue,
            urlcolor  = blue,
            citecolor = blue,
            anchorcolor = blue]{hyperref}
\usepackage{mathrsfs}
\usepackage{mathtools}
\usepackage{thmtools}
\usepackage{thm-restate}
\mathtoolsset{centercolon}
\usepackage[numbers,comma,sort&compress]{natbib}
\usepackage{bm}
\usepackage{bookmark}
\usepackage{svg}
\usepackage[font=small]{caption}
\usepackage{subcaption}
\usepackage{titlesec}
\usepackage{enumitem}

\usepackage[utf8]{inputenc}
\usepackage{color,soul}
\usepackage{graphicx, fullpage, dsfont}
\usepackage{authblk}

\usepackage{xcolor}
\usepackage{fullpage}
\usepackage{doi}
\usepackage{pgfplots}
\pgfplotsset{compat=newest}
\pgfplotsset{plot coordinates/math parser=false}
\newlength\figureheight
\newlength\figurewidth

\let\originalleft\left
\let\originalright\right
\renewcommand{\left}{\mathopen{}\mathclose\bgroup\originalleft}
\renewcommand{\right}{\aftergroup\egroup\originalright}

\newcommand{\R}{{\mathbb{R}}}

\renewcommand{\d}{\mathrm{d}}

\renewcommand{\Re}{\mathop{\mathrm{Re}}}

\newcommand{\norm}[1]{\|{#1}\|}

\newcommand{\range}[1]{[{#1}]}

\DeclareMathOperator{\poly}{poly}
\DeclareMathOperator{\polylog}{polylog}

\newtheoremstyle{tight}
  {1.5pt} 
  {1.5pt} 
  {\itshape} 
  {} 
  {\bfseries} 
  {.} 
  {.5em} 
  {} 
\theoremstyle{tight}
\newtheorem{theorem}{Theorem}
\newtheorem{lemma}{Lemma}

\newtheorem{problem}{Problem}
\newtheorem{remark}{Remark}
\newtheorem{definition}{Definition}

\newcommand{\eq}[1]{(\ref{eq:#1})}
\newcommand{\alg}[1]{\hyperref[alg:#1]{Algorithm~\ref*{alg:#1}}}
\newcommand{\defn}[1]{\hyperref[defn:#1]{Definition~\ref*{defn:#1}}}
\renewcommand{\sec}[1]{\hyperref[sec:#1]{Section~\ref*{sec:#1}}}
\newcommand{\thm}[1]{\hyperref[thm:#1]{Theorem~\ref*{thm:#1}}}
\newcommand{\lem}[1]{\hyperref[lem:#1]{Lemma~\ref*{lem:#1}}}
\newcommand{\cor}[1]{\hyperref[cor:#1]{Corollary~\ref*{cor:#1}}}
\newcommand{\prb}[1]{\hyperref[prb:#1]{Problem~\ref*{prb:#1}}}
\newcommand{\fig}[1]{\hyperref[fig:#1]{Figure~\ref*{fig:#1}}}
\newcommand{\appx}[1]{\hyperref[appx:#1]{Appendix~\ref*{appx:#1}}}

\makeatletter
\let\@@magyar@captionfix\relax
\makeatother

\numberwithin{equation}{section}

\newcommand{\be}{\begin{equation}}
\newcommand{\ee}{\end{equation}}

\begin{document}

\title{Improved quantum algorithms for linear and nonlinear differential equations}
\author{Hari Krovi}
\email{hari.krovi@riverlane.com, hkrovi@gmail.com}
\orcid{0000-0001-9675-9959}
\affil[]{ Riverlane Research, Cambridge, MA}

\maketitle

\begin{abstract}
We present substantially generalized and improved quantum algorithms over prior work for inhomogeneous linear and nonlinear ordinary differential equations (ODE). Specifically, we show how the norm of the matrix exponential characterizes the run time of quantum algorithms for linear ODEs opening the door to an application to a wider class of linear and nonlinear ODEs. In \cite{BCOW17}, a quantum algorithm for a certain class of linear ODEs is given, where the matrix involved needs to be diagonalizable. The quantum algorithm for linear ODEs presented here extends to many classes of non-diagonalizable matrices including singular matrices. The algorithm here is also exponentially faster than the bounds derived in \cite{BCOW17} for certain classes of diagonalizable matrices.

Our linear ODE algorithm is then applied to nonlinear differential equations using Carleman linearization (an approach taken recently by us in \cite{Liue2026805118}). The improvement over that result is two-fold. First, we obtain an exponentially better dependence on error. This kind of logarithmic dependence on error has also been achieved by \cite{Xue_2021}, but only for homogeneous nonlinear equations. Second, the present algorithm can handle any sparse matrix (that models dissipation) if it has a negative log-norm (including non-diagonalizable matrices), whereas \cite{Liue2026805118} and \cite{Xue_2021} additionally require normality.
\end{abstract}

\section{Introduction}
Differential equations lie at the heart of many important problems in several fields including fluid mechanics, plasma physics and quantum physics. Indeed, simulating the Schr\"{o}dinger equation, a central equation in quantum physics, was one of the first motivations to build a quantum computer and it is a linear homogeneous differential equation. Simulating the Schr\"{o}dinger equation or, in other words, Hamiltonian simulation can be done efficiently for several classes of Hamiltonians and important techniques were developed in a long line of work \cite{Lloyd_Ham_sim, BAC07, Ham_sim_opt, Low2019hamiltonian, chakraborty_et_al:LIPIcs:2019:10609}. The quantum algorithmic techniques developed in those works extend beyond Hamiltonian simulation. Some of those techniques such as block encoded matrix operations will be used in this paper. The technique of block encoded matrix operations is further developed in \cite{vanApeldoorn2020quantumsdpsolvers, 10.1145/3313276.3316366} by extending the block encoding to functions of a sparse matrix.

Another important development in quantum algorithms is the algorithm to solve linear systems of equations \cite{HHL09}. This algorithm has been applied to numerous problems such as in machine learning and differential equations \cite{Ber14, BCOW17,CLO20, CL19,Liue2026805118, LPG20}. The run-time of the original algorithm \cite{HHL09} has an exponentially better dependence on dimension and it scales quadratically as function of the condition number $\kappa$ of the linear system. This dependence on the condition number has been improved in subsequent work to $\kappa\log\kappa$ \cite{ambainis:LIPIcs:2012:3426}. In \cite{CKS15}, the dependence on the solution error has been exponentially improved and made logarithmic. In a recent series of papers, the scaling has been improved to $\kappa\log1/\epsilon$ \cite{Subasi,Dong_Lin,Lin2020optimalpolynomial,QLSA_linear_kappa}. The Quantum Linear Systems Algorithm (QLSA) with these improvements will be used here as well.

Quantum algorithms for differential equations have been designed in several papers starting with the early work on nonlinear differential equations in \cite{LO08}. In \cite{Ber14}, a quantum algorithm to solve linear inhomogeneous equations is presented that uses high-order methods. In \cite{BCOW17}, using truncated Taylor series, an exponential improvement in the dependence of the solution error was obtained. In \cite{CL19}, a quantum algorithm to solve time dependent linear differential equations was presented using spectral methods. These spectral methods are also applied to develop a quantum algorithm to solve linear partial differential equations in \cite{CLO20}.

Recently, in \cite{Liue2026805118}, a quantum algorithm to solve dissipative nonlinear differential equations is given using the so called Carleman linearization. This algorithm is efficient when the ratio of nonlinearity to dissipation is less than $1$. In \cite{LPG20}, a quantum algorithm for nonlinear differential equations is given when the time scale of simulation is short. However, this was heuristically derived without a rigorous analysis. Recently, using homotopic perturbation methods, another quantum algorithm for nonlinear differential equations is given in \cite{Xue_2021}, where the dependence on the error is exponentially better than in \cite{Liue2026805118}. However, the ratio of nonlinearity to dissipation is required to be smaller than in \cite{Liue2026805118}. The algorithm of \cite{Xue_2021} is also only for homogeneous differential equations with the assumption of normality. In the regime of nonlinearity considered in \cite{Xue_2021}, homogeneous differential equations have an exponential decay of solutions leading to a run-time that is exponential in $T$ (the simulation time).

Other works that address aspects of differential equations but with an application to plasma physics are \cite{ESP19}, where a quantum algorithm to solve a linearized Vlasov equation is given. In \cite{DS20, ESP21}, a formulation of several interesting problems in plasma physics that could be amenable to quantum algorithms through linearization is given and in \cite{Jos20}, a quantum algorithm to solve classical dynamics using the Koopman operator approach is given. In \cite{novikau2021quantum}, a quantum algorithm to solve cold plasma waves is given. In \cite{Hubisz21}, a connection between Wigner-Weisskopf theory of spontaneous decay and non-Hermitian evolution is made and an implementation using ``non-Hermitian" quantum circuits is explored.
 
In this paper, we first present a quantum algorithm to solve time-independent linear inhomogeneous equations that improves on the results of \cite{BCOW17} in two ways. In order to do this, we extend their algorithm to include many classes of non-diagonalizable and even singular matrices that are not covered by their algorithm. As was pointed out in \cite{BCOW17}, every matrix is close to a diagonalizable matrix (since the latter are dense in the set of all matrices). But using the closest diagonalizable matrix leads to an exponentially worse error. Here, we are able to extend to many classes of non-diagonalizable matrices keeping the logarithmic dependence on error. We give a characterization of matrices for which the algorithm is efficient by bounding the gate and query complexity. Diagonalization was needed in \cite{BCOW17} in deriving a bound on the condition number and a bound on the solution error. To circumvent the need for diagonalization, we provide a different and improved analyses for both these parts.

The second main difference between the two algorithms is that even when we restrict to diagonalizable matrices there are cases when the bound derived for the algorithm of \cite{BCOW17} is exponentially worse than ours. Specifically, for diagonalizable matrices that have non-positive log-norm but an exponentially high condition number $\kappa_V$, where $V$ is the matrix of eigenvectors, the gate complexity bound for the algorithm of \cite{BCOW17} is exponential since it has at least a linear dependence on $\kappa_V$. In our algorithm, this dependence is removed leading to an exponential improvement in this case. More generally, for all diagonalizable matrices with bounded norm of the matrix exponential (i.e., $\|\exp(At)\|$ is bounded) but with a high $\kappa_V$, our algorithm is exponentially better than the bound in \cite{BCOW17}. Numerical simulations for a specific example of the actual condition number of \cite{BCOW17} show that the bounds derived there might be very loose for such matrices.

Despite these differences, the algorithm presented here has many common elements with \cite{BCOW17}. Like that work, we use truncated Taylor series and construct a quantum linear system whose solution gives a quantum state proportional to the solution of the linear differential equation. We also make use of the ``ramp" to boost the success probability using the technique introduced in \cite{Ber14} and used in subsequent papers. However, the linear system is different in our case and we implement it using the techniques of block encoding. The reason we pick a different linear system is that it is easier to analyze using the techniques developed here.

We then apply our linear ODE algorithm to solve nonlinear differential equations by using Carleman linearization. We improve on our prior work \cite{Liue2026805118} by exponentially improving the dependence on the error. The main reason that \cite{Liue2026805118} has a polynomial dependence on the error is the use of Euler method. The linear ODE obtained by Carleman linearization need not be diagonalizable and hence \cite{BCOW17} is not immediately applicable. But the Carleman ODE has negative log-norm if the dissipation matrix of the nonlinear system also has that property. This makes it possible to use our present algorithm to achieve an exponential improvement in the error even though the matrix may not be diagonalizable. This also means that our linear ODE algorithm applied to Carleman linearized ODEs can handle non-normal and non-diagonalizable matrices. In principle, we should be able to generalize this further since all we need for our linear ODE algorithm is a general property that the norm of the matrix exponential is bounded. However, it seems hard to characterize dissipation matrices of nonlinear equations that lead to Carleman ODEs with this general property. 

Since the first version of this paper appeared, there have been some more results which use the characterization of run time in terms of matrix exponential introduced here. In \cite{reaction_diffusion, costa2023improving}, the authors apply Carleman linearization to reaction-diffusion equations which form a special class of nonlinear differential equations. In \cite{time_marching,berry2022quantum}, the authors design a quantum algorithm for time-dependent ODEs where the run-time bound depends on the matrix exponential and in \cite{jennings2023cost,jennings2023efficient} detailed run-time costs and optimizations of the algorithms for linear systems and differential equations is presented. In \cite{ode_theory}, some limitations of quantum algorithms for ODEs are explored. In \cite{jin2022quantum,LCHS,an2023quantum} the authors use a Hamiltonian simulation approach to non-unitary dynamics.

This paper is organized as follows. In \sec{problem_statement}, we precisely state the problem that our algorithm solves as well as some background on notation and results on block encoded matrices that we use later on in the paper. In \sec{norm_exp}, we discuss the behaviour of the norm of the matrix exponential, its dependence on quantities such as spectral abscissa and log-norm. We then give a characterization of matrices (even non-diagonalizable ones) that are well-behaved for our algorithm. Then in \sec{algorithm}, we present the steps of the algorithm. In \sec{analysis}, we present the analysis of the algorithm by deriving bounds on the solution error, condition number of the linear system, the probability of success and the implementation of the algorithm. In \sec{main}, we prove the correctness of the algorithm and bound its gate and query complexity. In \sec{nonlinear}, we apply the linear ODE algorithm to nonlinear differential equations and show an exponential improvement in error. We also show that more general classes of nonlinear differential equations can be solved efficiently than the ones considered earlier. Finally, in \sec{conclusions}, we present some conclusions and open questions.

\section{Problem statement and preliminaries}\label{sec:problem_statement}
In this section, we set notation and give some background and results from the literature on block encoding. We also precisely define the input model and the problem our algorithm solves.

First, we set notation and define the quantities we use in this paper. We also give results from prior work that we use later on in the paper. In this paper, $\|\cdot\|$ denotes $l_2$ norm of a vector or a matrix. For an arbitrary matrix $A$, an eigenvalue is a root of the characteristic polynomial and singular value is the square root of an eigenvalue of $A^\dag A$. We use the standard notation $g(n)=O(f(n))$ for a function $f(n)$ to mean $g(n)\leq c f(n)$ for some $c$ independent of $n$. To distinguish this $O$ from the $O$ used for oracles, we always use a subscript for the oracles (defined below).

The specific problem we consider is to obtain the solution of a linear ordinary differential equation as a quantum state. We state it below more precisely.
\begin{problem}\label{prb:ODE}
Assume we are given a stable, sparse $d\times d$ matrix $A$, consider the linear ordinary differential equation
\begin{equation}
    \frac{dx}{dt}=Ax + b\,,\,x(0)=x_0\,,
\end{equation}
 and access to $A$, $b$ and $x_0$ via oracles $O_A$, $O_b$ and $O_x$ as in \defn{oracles} below. Produce a quantum state proportional to $x(T)$ for some $T>0$ to within an error $\epsilon$ (in $l_2$ norm).
\end{problem}

The input model we consider is the same as the one considered in prior work on differential equations \cite{Ber14, BCOW17, CL19, CLO20, Liue2026805118}. The matrix $A$ is assumed to be sparse with at most $s_r$ nonzero entries in any row and at most $s_c$ nonzero entries in any column.
\begin{definition}\label{defn:oracles}
We assume that $A\in \mathbb{C}^{2^n\times 2^n}$ can be accessed through the following oracles.
\begin{align}
    &O_r\ket{i,k}=\ket{i,r_{ik}}\,,\\
    &O_c\ket{i,k}=\ket{i,c_{ik}}\,,\\
    &O_A\ket{i,j}\ket{0}^{\otimes b}=\ket{i,j,a_{ij}}\,,
\end{align}
where $r_{ij}$ (respectively $c_{ji}$) is the $j^{th}$ nonzero entry in the $i^{th}$ row (resp. column). If there are fewer than $j$ entries, then $r_{ij}$ (resp. $c_{ji}$) is set to $j+2^n$. Here $a_{ij}$ is a $b$-bit binary representation of the $(i,j)$ matrix entry of $A$.

Similarly, we have oracles $O_b$ and $O_x$ to prepare a normalized version of $b$ and $x_0$. We assume that their norms are known. Specifically, let $O_x$ be any unitary that maps $\ket{1}\ket{\phi}$ to $\ket{1}\ket{\phi}$ for any state $\ket{\phi}$ and $\ket{0}\ket{0}$ to $\ket{0}\ket{\bar{x}_{0}}$, where $\bar{x}_{0}=x_{0}/\|x_{0}\|$. Let $O_b$ be a unitary that similarly maps $\ket{1}\ket{\phi}$ to $\ket{1}\ket{\phi}$ for any state $\ket{\phi}$ and $\ket{0}\ket{0}$ to $\ket{0}\ket{\bar{b}}$, where $\bar{b}=b/\|b\|$.
\end{definition}

Next, we give here some results on block encoded matrices that will be useful for us. Our main tool to construct a quantum algorithm is to use matrix arithmetic to construct the linear system corresponding to a given ODE. The linear system can be broken up into pieces involving subtraction, multiplication and inversion of matrices and we give a list of results from the literature that bound the complexity of these operations.

First, we define block encoding. The idea behind block encoding is to embed an arbitrary sparse matrix into a unitary matrix so that one can use embedding to perform arithmetic operations (see \cite{10.1145/3313276.3316366} for more details).
\begin{definition}\label{defn:block_encoding}
Suppose $A$ is an $a$-qubit operator, $\alpha,\epsilon\in \mathbb{R}_+$ and $b\in \mathbb{N}$, then an $a+b$ qubit unitary $U$ is said to be an $(\alpha, b, \epsilon)$ block encoding of $A$ if
\begin{equation}
    \|A-\alpha(\bra{0}^b\otimes I)U(\ket{0}^b\otimes I)\|\leq \epsilon\,.
\end{equation}
\end{definition}
\begin{definition}\label{defn:Hermitian_complement}
We denote by $\bar{A}$ the Hermitian complement of $A$ defined below 
\begin{equation}
    \bar{A}=\begin{pmatrix}0&A\\A^\dag &0
    \end{pmatrix}\,.
\end{equation}
\end{definition}
Next, we state the relevant results from the literature that we use in our algorithm. First, we need the following lemma that describes a way to block-encode a sparse matrix.
\begin{lemma}[\cite{10.1145/3313276.3316366}, Lemma 48]\label{lem:sparse_matrix}
Suppose $A$ is an $s_r$-row-sparse and $s_c$-column-sparse $a$-qubit matrix where each element of $A$ has absolute value at most 1 and where we have oracle access to the positions and values of the nonzero entries, then we have an implementation of $(\sqrt{s_rs_c},a+3,\epsilon)$ block encoding of $A$ with a single use of $O_r$ and $O_c$ and two uses of $O_A$ and using $O(a+\log^{2.5}(s_rs_c/\epsilon))$ elementary gates and $O(b+\log^{2.5}(s_rs_c/\epsilon))$ ancillas.
\end{lemma}
Next, we need the following result from \cite{chakraborty_et_al:LIPIcs:2019:10609} to invert a block encoded matrix.
\begin{lemma}[\cite{chakraborty_et_al:LIPIcs:2019:10609}, Lemma 9]\label{lem:matrix_inversion}
Let $A$ be a matrix with condition number $\kappa\geq 2$. Let $H$ be the Hermitian complement of $A$ and suppose that $I/\kappa\leq H\leq I$. Let
\begin{equation}
    \delta = o(\epsilon/\kappa^2\log^3(\kappa^2/\epsilon))\,.
\end{equation}
If $U$ is an $(\alpha, a, \delta)$ block encoding of $H$ that has gate complexity $T_U$, then we can implement a 
\begin{equation}
    (2\kappa,a+O(\log(\kappa^2\log(1/\epsilon))),\epsilon)
\end{equation}
block encoding of $H^{-1}$ with gate complexity
\begin{equation}
    O(\alpha\kappa(a+T_U)\log^2(\kappa^2/\epsilon))\,.
\end{equation}
\end{lemma}
\begin{remark}
In the above lemma, when $H$ is a Hermitian complement of a matrix $A$, then after inversion, the resulting block encoded matrix satisfies
\begin{equation}
    \|A^{-1}-2\kappa(\bra{0}^b\otimes \bra{1}\otimes I)U(\ket{0}^b\otimes\ket{0}\otimes I)\|\leq \epsilon\,,
\end{equation}
where the middle qubit corresponds to the qubit needed for the Hermitian complement. By appending an $X$ gate to the block encoding for that qubit, this can be brought into the above form of block encoding of $A^{-1}$ (as in \defn{block_encoding}).
\end{remark}
The following result describes how to perform matrix arithmetic on block-encoded matrices.
\begin{lemma}[\cite{chakraborty_et_al:LIPIcs:2019:10609}, Lemma 6, \cite{10.1145/3313276.3316366}, Lemmas 52, 54]\label{lem:matrix_arithmetics}
If $A$ has an $(\alpha,a,\epsilon)$ block encoding with gate complexity $T_A$ and $B$ has a $(\beta,b,\delta)$ block encoding with gate complexity $T_B$, then
\begin{enumerate}
    \item $\bar{A}$ has an $(\alpha,a+1,\epsilon)$ block encoding that can be implemented with gate complexity $O(T_A)$.
    \item $A+B$ has an $(\alpha+\beta,a+b,\beta\epsilon+\alpha\delta)$ block encoding that can be implemented with gate complexity $O(T_A+T_B)$.
    \item $AB$ has an $(\alpha\beta,a+b,\alpha\delta+\beta\epsilon)$ block encoding with gate complexity $O(T_A+T_B)$.
\end{enumerate}
\end{lemma}

Next, we need the following theorem to implement the QLSA.
\begin{theorem}[\cite{QLSA_linear_kappa}, Theorem 19]\label{thm:block_QLSA}
Let $A$ be such that $\|A\|=1$ and $\|A^{-1}\|=\kappa$. Given a oracle block encoding of $A$ and an oracle for implementing $\ket{b}$, there exists a quantum algorithm which produces the normalized state $A^{-1}\ket{b}$ to within an error $\epsilon$ using 
\begin{equation}
    O(\kappa\log 1/\epsilon)
\end{equation}
calls to the oracles.
\end{theorem}
\begin{remark}
In using \lem{matrix_inversion} and \thm{block_QLSA}, we will apply them to matrices with norms that are $O(1)$ rather than $1$. This does not change the asymptotic scaling of the gate complexities.
\end{remark}
\section{Norm of the matrix exponential}\label{sec:norm_exp}
In the analysis of our algorithm, we will need to bound the norm of the matrix exponential $\|e^{At}\|$ for some matrix $A$ and $t\geq 0$. We explain the known results that bound this quantity and also the relation to stability theory of differential equations. For an arbitrary matrix $A$, in addition to the norm of the matrix, the following quantities are useful to bound the norm of functions of $A$.
\begin{definition}\label{defn:alpha_mu}
\begin{align}
    &\sigma(A)=\{\lambda\,|\, \lambda \text{ is an eigenvalue of }A\}&\text{Spectrum}\\
    &\alpha(A)=\max \{\Re(\lambda) \,|\, \lambda\in \sigma(A)\}&\text{Spectral abscissa} \\
    &\mu(A)=\max\{\lambda\,|\, \lambda\in \sigma((A+A^\dag)/2)\}&\text{Log-norm}\\
    &\rho(A)=\max\{|\lambda|\,|\, \lambda\in \sigma(A)\}&\text{Spectral radius}
\end{align}
\end{definition}
We drop the dependence on $A$ in the above quantities when it is clear from context. The quantity $\mu(A)$ is not a norm (even though it is called the log-norm) and can, in fact, be negative. The norm of $e^{At}$ depends on $\mu$ and $\alpha$ in the short and long timescales respectively. If we want the norm of $e^{At}$ to be bounded for all $t$, we need that $\alpha(A)<0$. But we will see below that this alone is not enough and we have to make additional assumptions depending on the sign of $\mu$. The quantities $\alpha$, $\mu$ and $\rho$ are all real numbers. For arbitrary matrices, these quantities are related as follows.
\begin{equation}
    \alpha \leq \mu\leq \|A\|\,\, \text{ and }\, \alpha\leq \rho\,.
\end{equation}
In the theory of differential equations \cite{coppel1965stability}, a \emph{stable} matrix is defined as a matrix whose eigenvalues (i.e., roots of the characteristic polynomial) have negative real parts. It is called \emph{semi-stable} if the eigenvalues are allowed to have zero real parts. The \emph{Lyapunov condition}, can be used to determine if a matrix is stable. This can be stated as follows. If there exist a positive definite matrix $P$ and a negative definite matrix $N$ such that $AP + PA^\dag=N$, then $A$ is stable. The Routh-Hurwitz criterion \cite{coppel1965stability} can also be used to check for stability. In this section, we are interested in bounds on the norm of the matrix exponential for stable matrices $A$. Next, we summarize known results on these bounds.

For an arbitrary matrix $A$, the exponential of $At$ for $t\geq 0$ converges to zero in the limit of long time if and only if $A$ is a stable matrix \cite{van2006study}, i.e.,
\begin{equation}
    \lim_{t\rightarrow \infty} \|\exp(At)\| = 0\, \, \iff \,\, A \text{ is stable i.e., }\alpha<0\,.
\end{equation}
The following bounds are well-known \cite{Dah63}.
\begin{equation}
    \exp(\alpha t)\leq \|\exp(At)\|\leq \exp(\mu t)\leq \exp(\|A\|t)\,.
\end{equation}
While this is a bound for all time, it turns out that the quantity $\|\exp(At)\|$ grows (or decays) like the function $\exp(\mu t)$ initially following it closely for very short times. In fact, another way to define the log norm is as the rate of change in $\|\exp(At)\|$ near $t=0$. This behavior can be seen in the following simple example \cite{trefethen2005spectra}. Consider the following two $2\times 2$ matrices.
\begin{equation}\label{eq:A_B}
    A=\begin{pmatrix}-2& 10\\0&-2\end{pmatrix}\,,\hspace{0.5in} B=\begin{pmatrix}-2& 1\\0&-2\end{pmatrix}\
\end{equation}
Both matrices have eigenvalues equal to $-2$ i.e., $\alpha<0$ for both, but $\mu(A)>0$ and $\mu(B)<0$. The plot in \fig{exp_norm_example} shows the difference in short term behavior of the norm of the matrix exponential depending on the sign of $\mu$.
\begin{figure}
    \centering
    \scalebox{0.7}{\input{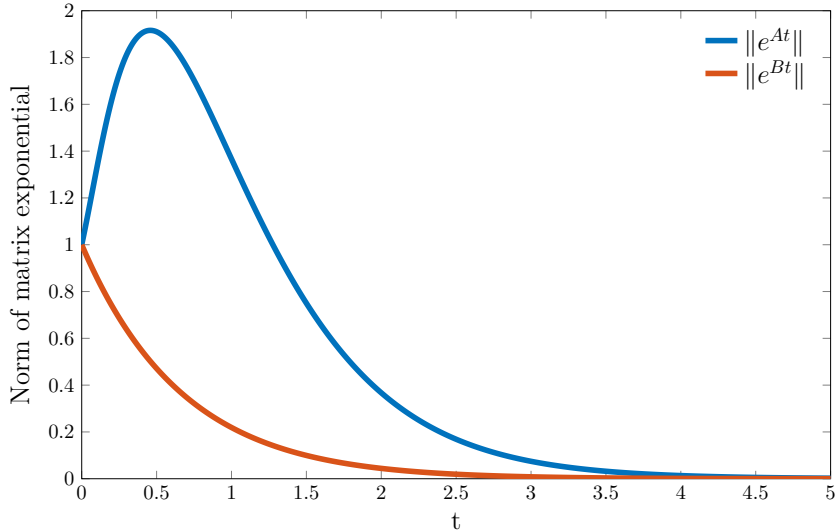}}
    \caption{Plot of $\|e^{At}\|$ and $\|e^{Bt}\|$ vs $t$, where $A$ and $B$ are given in \eq{A_B}. The blue curve corresponds to the matrix $A$ which has $\mu>0$ and the red curve to the matrix $B$ which has $\mu<0$.}
    \label{fig:exp_norm_example}
\end{figure}

We would like to point out that restricting attention to upper triangular matrices is without loss of generality since by the Schur decomposition \cite{bhatia1996matrix}, any matrix is unitarily equivalent to an upper triangular matrix and the $l_2$ norm is unitarily invariant. Based on the above example, we can consider the two cases depending on the log-norm: when $\mu\leq 0$ and $\mu>0$.

\subsection{Non-positive log-norm}
In this case, the matrix exponential is easily bounded using the following inequality.
\begin{equation}\label{eq:mu_bound}
    \|\exp(At)\|\leq e^{\mu t}\leq 1\,.
\end{equation}
This bound in useful when the log-norm is known to be negative. For instance, consider the following tridiagonal matrix.
\begin{equation}\label{eq:A}
    A=\frac{1}{d}\begin{pmatrix}
    -1 & i & 0 & \hdots & 0 & 0\\
    i & -2 & 2i & \hdots & 0 & 0\\
    0 & 2i & -3 & \hdots & 0 & 0\\
    \vdots & \vdots&\vdots&\ddots &\vdots&\vdots\\
    0 & 0 & 0 &\hdots & -(d-1) & (d-1)i\\
    0 & 0 & 0 &\hdots & (d-1)i & -d
    \end{pmatrix}\,.
\end{equation}
Matrices of this type are known as twisted Toeplitz or Berezin-Toeplitz matrices. Matrices like these appear in certain plasma physics models \cite{Lou15}. Matrices with a negative log-norm appear in high energy physics e.g., in modeling particle decay \cite{Hubisz21, Particle_decay}. The above matrix is diagonalizable but is not a normal matrix. Therefore, the diagonalizing matrix $V$ is not unitary. In this case, $V$ has a condition number $\kappa_V$ that is exponential in $\log d$ (see \fig{comparison}). However, it can be seen that $(A+A^\dag)/2$ is a diagonal matrix with negative entries and therefore $\mu<0$. In fact, it can be seen that $\mu=-1/d$. This also means that $\alpha<0$ since $\alpha<\mu$. If $\kappa_V$ grows exponentially as in this case (see \fig{comparison}), the bound
\begin{equation}
    \|\exp(At)\|\leq \kappa_V \exp(\alpha t)\,,
\end{equation}
is exponentially worse than the bound \eq{mu_bound}.


\begin{figure}
     \centering
     \begin{subfigure}[b]{0.49\textwidth}
         \centering
         \scalebox{0.55}{
%
%
\definecolor{mycolor1}{rgb}{0.00000,0.44700,0.74100}%
\definecolor{mycolor2}{rgb}{0.85000,0.32500,0.09800}%
\definecolor{mycolor3}{rgb}{0.92900,0.69400,0.12500}%
\begin{tikzpicture}

\begin{axis}[%
width=4.528in,
height=3.49in,
at={(2.511in,1.906in)},
scale only axis,
xmin=15,
xmax=100,
xlabel style={font=\color{white!15!black}},
xlabel={Dimension of the matrix A},
ymode=log,
ymin=1,
ymax=587523094872436,
yminorticks=true,
ylabel style={font=\color{white!15!black}},
ylabel={Condition numbers and bounds},
axis background/.style={fill=white},
xmajorgrids,
ymajorgrids,
yminorgrids,
legend style={legend cell align=left, align=left, fill=none, draw=none}
]
\addplot [color=mycolor1, line width=3.0pt]
  table[row sep=crcr]{%
10	38.0607720885288\\
12	38.9531633371435\\
14	39.5793828751334\\
16	40.041588546251\\
18	40.395840304143\\
20	40.6764561604494\\
22	40.9026988984043\\
24	41.0892360278268\\
26	41.2455631264072\\
28	41.3783932083567\\
30	41.492601108138\\
32	41.5925589751258\\
34	41.6795534906461\\
36	41.756408099356\\
38	41.8247827865856\\
40	41.8859960145504\\
42	41.9411078992081\\
44	41.9909802823533\\
46	42.0368630921021\\
48	42.0782722620973\\
50	42.1162244114355\\
52	42.1511320696512\\
54	42.183344952426\\
56	42.2131614379148\\
58	52.0563562214347\\
60	52.0955908834392\\
62	52.1322020087108\\
64	52.1664426541675\\
66	52.1985344826102\\
68	52.2292857711667\\
70	52.2576481874251\\
72	52.2843808958607\\
74	52.3096196688122\\
76	52.3334856626672\\
78	52.3560873275285\\
80	52.3775220257385\\
82	52.397877409668\\
84	52.4172325994714\\
86	52.4356591938305\\
88	52.4532221406322\\
90	52.4699804896807\\
92	52.4859880456493\\
94	52.5012939363354\\
96	52.5159431087533\\
98	52.5304249990391\\
100	52.5438821164605\\
};
\addlegendentry{$\kappa{}_\text{L}$}

\addplot [color=mycolor2, line width=3.0pt]
  table[row sep=crcr]{%
10	49.1966591059611\\
12	50.1187676291374\\
14	59.9115269023128\\
16	60.4731126868349\\
18	60.9036861704982\\
20	61.244080722093\\
22	61.5190326003002\\
24	61.7457561257756\\
26	61.9357851708564\\
28	71.5578882183024\\
30	71.7169181902335\\
32	71.8551885323373\\
34	71.9762419715353\\
36	72.0831913171923\\
38	72.1783476960641\\
40	81.811072684157\\
42	81.8974229077151\\
44	81.9755717061187\\
46	82.0465750509564\\
48	82.1114128838389\\
50	82.1708420579035\\
52	91.818519508255\\
54	91.8745646747555\\
56	91.9264448454719\\
58	113.458076923363\\
60	113.529867551094\\
62	113.59686301892\\
64	113.659526234093\\
66	113.718262619853\\
68	125.657630188751\\
70	125.71473426971\\
72	125.768560050236\\
74	125.819380869077\\
76	125.867440612545\\
78	125.912957565724\\
80	125.956127676153\\
82	137.911836164867\\
84	137.954380237421\\
86	137.994886064801\\
88	138.033496281043\\
90	138.070340556611\\
92	138.105537034627\\
94	138.139193580708\\
96	138.171408873955\\
98	150.141168382894\\
100	150.17320815528\\
};
\addlegendentry{$\kappa{}_\text{C}$}

\addplot [color=mycolor3, line width=3.0pt]
  table[row sep=crcr]{%
10	17.5352873756155\\
12	34.8261598350598\\
14	67.6775127475753\\
16	136.249461184391\\
18	309.828339089228\\
20	962.241221157035\\
22	1697.46202633565\\
24	2327.11057704636\\
26	4450.98797296117\\
28	8994.90682698602\\
30	17751.3331687894\\
32	36465.1803634661\\
34	84257.0282050832\\
36	262745.21049764\\
38	478064.541844091\\
40	656670.913763389\\
42	1262216.53833562\\
44	2561310.13680063\\
46	5076618.77458038\\
48	10480261.8170496\\
50	24302637.0004239\\
52	75306715.9100153\\
54	141688655.694126\\
56	193017057.291063\\
58	371208120.754903\\
60	754755078.474247\\
62	1499110861.59545\\
64	3099464297.81528\\
66	7190124193.18112\\
68	22084578476.5426\\
70	43003929167.0105\\
72	57880471065.0146\\
74	111231040608.656\\
76	226460116312.83\\
78	450446032266.413\\
80	930941973041.298\\
82	2156586157880.43\\
84	6663210939850.71\\
86	12946669170348.6\\
88	17583337392717.4\\
90	36344461788939\\
92	78551000295515.8\\
94	116852785073756\\
96	255576368857939\\
98	364610699997235\\
100	587523094872436\\
};
\addlegendentry{$\kappa{}_\text{V}$}

\end{axis}
\end{tikzpicture}%
}
         \caption{Plot comparing $\kappa_L$, $\kappa_C$ and $\kappa_V$.}
         \label{fig:comparison}
     \end{subfigure}
     \begin{subfigure}[b]{0.49\textwidth}
         \centering
         \scalebox{0.55}{
%
%
\definecolor{mycolor1}{rgb}{0.00000,0.44700,0.74100}%
\definecolor{mycolor2}{rgb}{0.85000,0.32500,0.09800}%
\begin{tikzpicture}

\begin{axis}[%
width=4.764in,
height=3.49in,
at={(2.275in,1.906in)},
scale only axis,
xmin=15,
xmax=100,
xlabel style={font=\color{white!15!black}},
xlabel={Dimension of the matrix A},
ymode=log,
ymin=39.8104857106922,
ymax=150.17320815528,
yminorticks=true,
ylabel style={font=\color{white!15!black}},
ylabel={Condition numbers of linear systems},
axis background/.style={fill=white},
xmajorgrids,
ymajorgrids,
yminorgrids,
legend style={legend cell align=left, align=left, fill=none, draw=none}
]
\addplot [color=mycolor1, line width=3.0pt]
  table[row sep=crcr]{%
10	38.0607720885288\\
12	38.9531633371435\\
14	39.5793828751334\\
16	40.041588546251\\
18	40.395840304143\\
20	40.6764561604494\\
22	40.9026988984043\\
24	41.0892360278268\\
26	41.2455631264072\\
28	41.3783932083567\\
30	41.492601108138\\
32	41.5925589751258\\
34	41.6795534906461\\
36	41.756408099356\\
38	41.8247827865856\\
40	41.8859960145504\\
42	41.9411078992081\\
44	41.9909802823533\\
46	42.0368630921021\\
48	42.0782722620973\\
50	42.1162244114355\\
52	42.1511320696512\\
54	42.183344952426\\
56	42.2131614379148\\
58	52.0563562214347\\
60	52.0955908834392\\
62	52.1322020087108\\
64	52.1664426541675\\
66	52.1985344826102\\
68	52.2292857711667\\
70	52.2576481874251\\
72	52.2843808958607\\
74	52.3096196688122\\
76	52.3334856626672\\
78	52.3560873275285\\
80	52.3775220257385\\
82	52.397877409668\\
84	52.4172325994714\\
86	52.4356591938305\\
88	52.4532221406322\\
90	52.4699804896807\\
92	52.4859880456493\\
94	52.5012939363354\\
96	52.5159431087533\\
98	52.5304249990391\\
100	52.5438821164605\\
};
\addlegendentry{$\kappa{}_\text{L}$}

\addplot [color=mycolor2, line width=3.0pt]
  table[row sep=crcr]{%
10	49.1966591059611\\
12	50.1187676291374\\
14	59.9115269023128\\
16	60.4731126868349\\
18	60.9036861704982\\
20	61.244080722093\\
22	61.5190326003002\\
24	61.7457561257756\\
26	61.9357851708564\\
28	71.5578882183024\\
30	71.7169181902335\\
32	71.8551885323373\\
34	71.9762419715353\\
36	72.0831913171923\\
38	72.1783476960641\\
40	81.811072684157\\
42	81.8974229077151\\
44	81.9755717061187\\
46	82.0465750509564\\
48	82.1114128838389\\
50	82.1708420579035\\
52	91.818519508255\\
54	91.8745646747555\\
56	91.9264448454719\\
58	113.458076923363\\
60	113.529867551094\\
62	113.59686301892\\
64	113.659526234093\\
66	113.718262619853\\
68	125.657630188751\\
70	125.71473426971\\
72	125.768560050236\\
74	125.819380869077\\
76	125.867440612545\\
78	125.912957565724\\
80	125.956127676153\\
82	137.911836164867\\
84	137.954380237421\\
86	137.994886064801\\
88	138.033496281043\\
90	138.070340556611\\
92	138.105537034627\\
94	138.139193580708\\
96	138.171408873955\\
98	150.141168382894\\
100	150.17320815528\\
};
\addlegendentry{$\kappa{}_\text{C}$}

\end{axis}
\end{tikzpicture}%
}
         \caption{Plot showing the gap between $\kappa_L$ and $\kappa_C$.}
         \label{fig:comparison_b}
     \end{subfigure}
        \caption{Comparison of the condition number $\kappa_L$ of the linear system $L$ defined here, the condition number $\kappa_C$ of the linear system $C$ defined in \cite{BCOW17} and $\kappa_V$, the condition number of the matrix that diagonalizes $A$. (a) This plot indicates that the bound based on $\kappa_V$ is exponentially greater than both the actual condition number $\kappa_C$ and $\kappa_L$. (b)
        This plot indicates that the gap between $\kappa_C$ and $\kappa_L$ grows with the dimension $d$ of the matrix $A$ (although it seems to be only polynomial in $\log d$).}
        \label{fig:comparison_plots}
\end{figure}
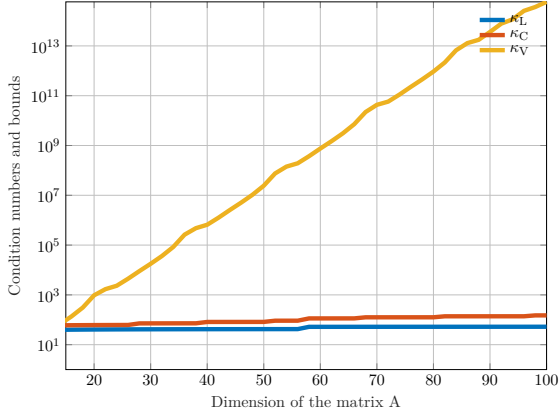
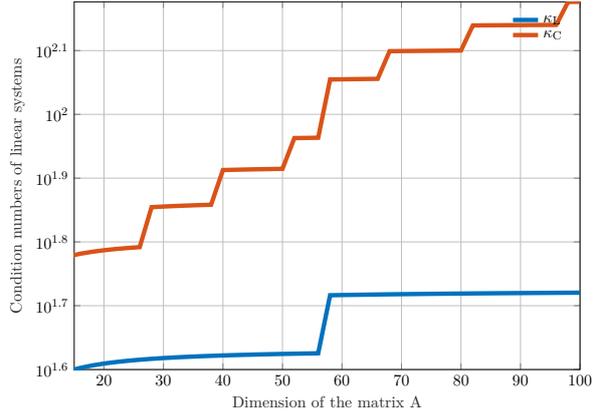

Using this example, we check here numerically if the condition number of the linear system from \cite{BCOW17} is indeed exponential or if the bound is just loose. The matrix $C$ in the linear system from \cite{BCOW17} is defined as
\begin{align}
    C =& \sum_{j=0}^d \ket{j}\bra{j}\otimes I - \sum_{i=0}^{m-1}\sum_{j=1}^k \ket{i(k+1)+j}\bra{i(k+1)+j-1}\otimes A/j \nonumber \\
    & - \sum_{i=0}^{m-1}\sum_{j=0}^k \ket{(i+1)(k+1)}\bra{i(k+1)+j}\otimes I - \sum_{j=d-p+1}^d \ket{j}\bra{j=1}\otimes I\,,
\end{align}
where $d=m(k+1)+p$.

We compare the actual condition number of the linear system $C$ in \cite{BCOW17} to the bound derived in \cite{BCOW17} as well as the condition number of the linear system $L$ defined here. In \fig{comparison}, we plot the relevant quantities for the matrix in \eq{A} for dimensions ranging from $15$ to $100$. We find that the condition number $\kappa_C$ of the linear system $C$ is far smaller than the bound derived in \cite{BCOW17}. While the bound, which depends on the condition number $\kappa_V$ of the diagonalizing matrix $V$, turns out to be exponential, $\kappa_C$ seems to scale polynomially for this example. The condition number $\kappa_L$ of the linear system $L$ that we define here is lower than $\kappa_C$ and the gap between them grows with dimension as shown in \fig{comparison_b}. It is possible that for the matrix in \eq{A} as well as other matrices $A$, the condition number $\kappa_C$ scales polynomially if the norm of the exponential of $A$ is bounded. If this is true, then it seems likely that the techniques developed in the present work can be used to prove this.

\subsection{Positive log-norm}
 For the case of positive log-norm, it turns out that $\|\exp(At)\|$ can increase initially since the log norm is positive. If we assume that $\alpha<0$, the long time behavior is exponential decay. The initial increase in norm is bounded and the Kreiss matrix theorem gives a bound on its supremum.
 
The Kreiss matrix theorem states the following.
 \begin{theorem}[Kreiss matrix theorem \cite{trefethen2005spectra}]
 For an arbitrary $d\times d$ matrix $A$, the following bounds are optimal.
 \begin{equation}
     \mathcal{K}(A)\leq \sup_{t\geq 0}\|\exp(At)\|\leq ed\mathcal{K}(A)\,,
 \end{equation}
 where $e$ is the Euler constant, $d$ is the dimension of $A$ and the Kreiss constant $\mathcal{K}(A)$ is defined as
 \begin{equation}
     \mathcal{K}(A)=\sup_{\Re(z)>0}\Re(z)\|(z I-A)^{-1}\|\,.
 \end{equation}
 \end{theorem}
 While this bound is useful for arbitrary matrices, the fact that, as a uniform bound, it is optimal and can be as large as the dimension of $A$ shows that we need to restrict the class of matrices to focus on.
 
 To restrict the class of matrices, we give a couple of bounds on $\|\exp(At) \|$. These bounds and others based on pre-conditioning are discussed in \cite{kaagstrom1977bounds}. The two main ways to derive bounds on the matrix exponential are to use the Jordan and Schur decompositions \cite{bhatia1996matrix}. The following bound using Jordan blocks was shown in \cite{van2006study}.
 
 \begin{lemma}\label{lem:exp_norm_1}
For a matrix $A$, let $A=VJV^{-1}$ be its Jordan decomposition and let $\beta$ be the dimension of the largest Jordan block. For $t\geq 0$, we have
\begin{equation}
    \|\exp(At)\|\leq \kappa(V)\exp(\alpha t)\beta\max_{0\leq r\leq \beta-1}\frac{t^r}{r!}\,.
\end{equation}
\end{lemma}

\begin{remark}
Note that when all the Jordan blocks are one dimensional i.e., the diagonalizable case, this reduces to the bound from \cite{BCOW17}.
\end{remark}

The next bound was also shown in \cite{van2006study} and it uses the Schur decomposition.
\begin{lemma}\label{lem:exp_norm_2}
Suppose the Schur decomposition of a matrix $A$ is $A=U(D+N)U^\dag$, where $U$ is a unitary matrix, $D$ is a diagonal matrix containing the eigenvalues of $A$ and $N$ is a nilpotent matrix that is strictly upper triangular. Since this decomposition is not unique let $N$ be such that $\|N\|$ is minimized. Then
\begin{equation}
    \|\exp(At)\|\leq p_{d-1}(\|N\|t)\exp(\alpha t)\,,
\end{equation}
where $d$ is the dimension of $A$ and
\begin{equation}
    p_{d-1}(x)=\sum_{j=0}^{d-1} \frac{x^j}{j!}\,.
\end{equation}
\end{lemma}

\begin{remark}
For $\sup_{t\geq 0}\|\exp(At)\|$ to be polynomial in $n=\log d$, using \lem{exp_norm_1} we would need the dimension $\beta$ of the largest Jordan block to be $\poly(n)$ (as well as a similar bound on $\kappa_V$). If we use \lem{exp_norm_2}, we need $\|N\|$ to bounded. Both of these conditions can be thought of as bounds on the non-normality of $A$. The quantity $\|N\|$ is known as \emph{departure from normality}. More measures of non-normality and their relationships are explored in \cite{ELSNER1987107}.
\end{remark}

The reason we present several bounds here is that no one bound completely subsumes the others, especially when the log-norm is positive. When $\mu(A)>0$, it is possible that the bound of \eq{mu_bound} is better than the ones in \lem{exp_norm_1} and \lem{exp_norm_2} for some $t$. This is because for the latter bounds, the exponential $\exp(\alpha t)$ needs to overcome the constants $\kappa_V$ and polynomials $p_{d-1}(\|N\|t)$ or $t^\beta$ before falling below a reasonable value, whereas for small positive $\mu$ and $t$, the former could be small.

Finally, we define here a quantity that we use in our algorithm that may be satisfied by matrices in the above classes e.g., in the remark above and more.
\begin{definition}\label{defn:C(A)}
\begin{equation}
    C(A)=\sup_{t\in [0,T]} \|\exp(At)\|\,.
\end{equation}
\end{definition}
The reason we are interested in this definition is that it encompasses even non-stable matrices. For instance, if $T\|A\|$ is a constant (independent of $n=\log d$), then $C(A)\leq \exp(T\|A\|)$ is also acceptable as a bound where $A$ can be an arbitrary matrix. To verify that $C(A)$ is bounded is not always an easy task. This is not uncommon in quantum or classical algorithms. For instance, verifying that the condition number is bounded for a linear system is not generally easy or checking if a matrix is diagonalizable as in the some prior algorithms. However, in some cases as in the example given above in \eq{A}, it is easy to verify that the matrix has a non-positive log norm. In general, when $A$ is a block matrix such as in the application for non-linear differential equations, one can bound $C(A)$ without additional assumptions on the component matrices such as e.g., normality or diagonalizability.

\section{Description of the quantum algorithm}\label{sec:algorithm}
We assume that we have a time-independent linear system of the form.
\be\label{eq:lode}
\frac{dx}{dt}=Ax + b\,,\, x(0)=x_0\,,
\ee
where $A$ and $b$ are time independent. The solution $x(T)$ at time $T$ of this matrix differential equation is
\begin{equation}\label{eq:soln}
    x(T)=\exp(AT)x_0 + T\sum_{j=0}^\infty \frac{(AT)^j}{(j+1)!}b\,,
\end{equation}
which is valid even when $A$ is a singular matrix \cite{higham2008functions}. We can re-write it in a form as follows that will be useful later.
\begin{lemma}\label{lem:soln_form}
We can write the solution of the linear ODE \eq{lode} as 
\begin{equation}
    x(T)=\exp(AT)x_0 + \Big[\int_0^T \exp(As) ds\Big] b\,.
\end{equation}
\end{lemma}
\begin{proof}
Expanding the exponential, we have for the second term
\begin{equation}
    \Big[\int_0^T \exp(As) d s\Big] b = \Big[\sum_{j=0}^\infty \int_0^T \frac{(As)^j}{j!} d s\Big]b =\sum_{j=0}^\infty \frac{A^j T^{j+1}}{(j+1)!}b\,.
\end{equation}
This can be seen to be equal to \eq{soln}. 
\end{proof}

To design a quantum algorithm to find the solution at time $T$, we use the method from \cite{BCOW17} of truncated Taylor series. Using their notation, we get the following discrete equation at time step $i$.
\be\label{eq:y_iteration}
y_{i+1}=T_k(Ah)y_i + S_k(Ah)hb\,,
\ee 
where $h$ is the time interval of a single step and $T_k(z)$ and $S_k(z)$ are the following functions. 
\begin{align}
    &T_k(z)=\sum_{j=0}^k \frac{z^j}{j!}\,,\\
    &S_k(z)=\sum_{j=1}^k\frac{z^{j-1}}{j!}\,.
\end{align}
Here we have $m=\lceil T/ h\rceil$ time steps in the time interval $[0,T]$. We will encode this equation into a linear system and use the quantum linear system algorithm (QLSA) to solve it. Our implementation also approximates via a truncated Taylor series similar to \cite{BCOW17}, but we use a different linear operator that we find easier to analyze to provide a better analysis and generalization. 

To solve the linear inhomogeneous equation, we use the linear operator $L$, where
\begin{align}\label{eq:linear_system}
    &L=I-N\\
    &N=\sum_{i=0}^m\ket{i+1}\bra{i}\otimes M_2(I-M_1)^{-1} + \sum_{i=m+1}^{m+p-1}\ket{i+1}\bra{i}\otimes I\\
    &M_1=\sum_{j=0}^{k-1} \ket{j+1}\bra{j}\otimes \frac{Ah}{j+1} \\
    &M_2=\sum_{j=0}^k \ket{0}\bra{j}\otimes I\,.
\end{align}

Here the first register, indexed by $i$, in the sum contains the time step. The second register, indexed by $j$, is the Taylor step and the third register contains states on which $A$ acts. The quantity $p$ is the number of additional steps needed to boost the success probability (this has become a standard trick that was first used in \cite{Ber14} and subsequently in several other papers \cite{CL19, Liue2026805118, Xue_2021}). We give an example of $M_1$ and $M_2$ for $k=3$ below.
\begin{equation}
    M_1=\begin{pmatrix}
      0 &0 &0&0\\
      Ah&0&0&0\\
      0&\frac{Ah}{2}&0&0\\
      0&0&\frac{Ah}{3}&0
    \end{pmatrix}\,\hspace{0.2in} \text{and}\hspace{0.2in}
    M_2=\begin{pmatrix}
      I &I &I&I\\
      0&0&0&0\\
      0&0&0&0\\
      0&0&0&0
    \end{pmatrix}\,.
\end{equation}

The above operator $L$ is used to solve the following linear system.
\begin{equation}
    L\ket{y}=\ket{\psi_{in}}\,,
\end{equation}
where 
\begin{equation}\label{eq:initial_state}
    \ket{\psi_{in}}=\ket{0,0,x_0} + h\sum_{i=0}^{m-1}\ket{i,1,b}\,,
\end{equation}
is the unnormalized version of the initial state. We now show that implementing $L^{-1}$ on the initial state gives us an entangled state from which we can extract the solution. Applying $L^{-1}$ to the initial state gives
\begin{equation}
    L^{-1}\ket{\psi_{in}} = \sum_{j=0}^{m+p-1} N^j \ket{\psi_{in}}\,,
\end{equation}
since $N$ is nilpotent and $N^{m+p+1}=0$.

To examine the above state, let us consider the action of $(I-M_1)^{-1}$ on an arbitrary basis state of the form $\ket{\ell,v}$. We have
\begin{align}
    (I-M_1)^{-1}\ket{\ell,v}&=\sum_{j=0}^{k}M_1^j\ket{\ell,v}\\
    &=\sum_{j=0}^{k-\ell}\ket{\ell+j}\otimes \frac{\ell!(Ah)^j}{(\ell+j)!}\ket{v}\,,
\end{align}
since $M_1$ is nilpotent and $M_1^{k+1}=0$. The action of $M_2$ on this state is to sum the entries giving us
\begin{equation}\label{eq:M_1M_2action}
    M_2(I-M_1)^{-1}\ket{\ell,v}=\ket{0}\otimes T_{\ell,k}(Ah)\ket{v}\,,
\end{equation}    
where
\begin{equation}
    T_{\ell,k}(Ah)=\sum_{j=0}^{k-\ell}\frac{\ell!(Ah)^j}{(\ell+j)!}\,.
\end{equation}
Note that $T_{0,k}(Ah)=T_k(Ah)$. Therefore, the action of this on the initial state is
\begin{equation}
    M_2(I-M_1)^{-1}(\ket{0,x_0} + h\ket{1,b})=\ket{0}\otimes\Big(\sum_{j=0}^k \frac{(Ah)^j}{j!}\ket{x_0} + \sum_{j=0}^{k-1} h\frac{(Ah)^{j}}{(j+1)!}\ket{b}\Big)=\ket{0,y_1}\,,
\end{equation}
where $y_i$ are given by \eq{y_iteration}.
This becomes the initial state for the next power of $N$ giving us
\begin{equation}
    N^j\ket{\psi_{in}} = \ket{j}\otimes \ket{0,y_{j}}\,.
\end{equation}
The overall state is then
\begin{equation}\label{eq:final_state}
   \ket{y} = L^{-1}\ket{\psi_{in}}=\sum_{i=0}^{m}\ket{i,0,y_i}+\sum_{i=m+1}^{m+p-1}\ket{i,0,y_m}\,.
\end{equation}
Now if one measures the time-step register and post-selects on the outcomes $m,\dots, m+p-1$, the state collapses to a state $\ket{\psi}$ which is $y_m/\|y_m\|$. We show in subsequent sections that this state is $\epsilon$ close to the normalized solution, that the measurement probability is significant, that the condition number is well-behaved and give bounds on the query and gate complexities of implementing the algorithm.

\section{Analysis of the quantum algorithm}\label{sec:analysis}
In this section, we first derive bounds on the solution error and then bound the condition number and probability of success of the algorithm.

For the solution error, we are interested in deriving a bound on the number of terms to keep for the Taylor series for the error between the actual and approximate solutions of the ODE to be small. In \cite{BCOW17}, the same thing was done assuming that the matrix is diagonalizable. The results here are valid for an arbitrary sparse matrix which can be non-diagonalizable and even singular. Standard books such as \cite{hairer2008solving} do not seem to have the specific results we need.
\subsection{Solution error}
Before we prove our bounds, let us define the following quantities. From \lem{soln_form}, the solution after evolving for time $T$ can be written as 
\begin{equation}
    x_T=\exp(AT)x_0+\Big[\int_0^T \exp(As)d s\Big]b\,,
\end{equation}
We also define
\begin{align}
    &l_0=\exp(Ah)\,,\\
    &l'_0=T_k(Ah)\,,\\
    &l_1=h\sum_{j=0}^\infty \frac{(Ah)^j}{(j+1)!}\,,\\
    &l_1'=h\sum_{j=0}^k \frac{(Ah)^j}{(j+1)!}\,.
\end{align}
and the remainder term
\begin{equation}
    R_k(Ah)=l_0-l'_0=\sum_{j=k+1}^\infty \frac{(Ah)^j}{j!}\,.
\end{equation}
We also define $L_0=l_0^m$, $L'_0=l_0^{'m}$, and 
\begin{align}
    &L_1=T\sum_{j=0}^\infty \frac{(AT)^j}{(j+1)!}\,,\\
    &L'_1=h\sum_{j=0}^{m-1}T_k^j(Ah)S_k(Ah)\,,
\end{align}
After $m$ time steps, the solution is
\begin{equation}
    y_m=L_0' x_0 + L_1'b\,.
\end{equation}
We need the following bounds later.
\begin{lemma}\label{lem:l_0}
When $\|Ah\|\leq 1$, we have
\begin{equation}
    \|(l_0-l'_0)l_0^{-1}\|\leq \frac{e^2}{(k+1)!}\,.
\end{equation}
\end{lemma}
\begin{proof}
In terms of the remainder, we have
\begin{align}
    \|(l_0-l'_0)l_0^{-1}\| &\leq \norm{R_k(Ah)\exp(-Ah)} \\
    &=\Big\| \frac{(Ah)^{k+1}}{(k+1)!} \sum_{j=0}^\infty \frac{(Ah)^j}{j!}\frac{\exp(-Ah)}{\binom{k+j+1}{j}}\Big\|\\
    &\leq \Big\| \frac{(Ah)^{k+1}}{(k+1)!}\Big\|\Big \| \sum_{j=0}^\infty \frac{(Ah)^j}{j!}\frac{1}{\binom{k+j+1}{j}}\Big\|\Big\|\exp(-Ah)\Big\|\\
    &\leq \frac{1}{(k+1)!}\sum_{j=0}^\infty\frac{\|Ah\|^j}{j!}\|\exp(-Ah)\|\\
    &\leq \frac{e^2}{(k+1)!}\,,
\end{align}
since $\norm{Ah}\leq 1$. 
\end{proof}
We have a similar bound on $l_1$ and $l_1'$.
\begin{lemma}\label{lem:l_1}
When $\|Ah\|\leq 1$, we have
\begin{equation}
    \|l_1 - l_1'\|\leq \frac{e}{\|A\|(k+1)!}\,.
\end{equation}
\end{lemma}
\begin{proof}
We have
\begin{equation}
    \|l_1-l_1'\|=\|h\sum_{j=k+1}^\infty \frac{(Ah)^j}{(j+1)!}\|\leq \frac{h}{(k+1)!}\|\sum_{j=0}^\infty\frac{1}{j!}\|\leq \frac{e}{\|A\|(k+1)!}\,.
\end{equation}
\end{proof}

\begin{lemma}\label{lem:L_0}
When $\|Ah\|\leq 1$ and $me^2/(k+1)!\leq 1$,
\begin{equation}
    \|(L_0-L'_0)L_0^{-1}\|\leq \frac{(e-1)me^2}{(k+1)!}\,.
\end{equation}
\end{lemma}
\begin{proof}
Since $l_0$ and $l'_0$ commute, we have
\begin{equation}
    \|(L_0-L'_0)L_0^{-1}\|=\|I-(l'_0l_0^{-1})^m\|=\|(I-l'_0l_0^{-1})(I+l'_0l_0^{-1} + \dots (l'_0l_0^{-1})^{m-1})\|\,.
\end{equation}
Using \lem{l_0} and the triangle inequality, we have
\begin{equation}
    \|l'_0l_0^{-1}\|\leq 1+ \|l'_0l_0^{-1}-I\| \leq 1+\frac{e^2}{(k+1)!}\,.
\end{equation}
Denote $r=1+\frac{e^2}{(k+1)!}$. Using this, we get
\begin{align}
    \|(L_0-L'_0)L_0^{-1}\|&\leq \|I-l'_0l_0^{-1}\|(1+r+\dots r^{m-1})\\
    &\leq \frac{e^2}{(k+1)!}\frac{r^m-1}{r-1}\\
    &=\Big(1+\frac{e^2}{(k+1)!}\Big)^m-1\\
    &=\sum_{j=1}^m \binom{m}{j} \Big(\frac{e^2}{(k+1)!}\Big)^j\\
    &=\sum_{j=1}^m \frac{1}{m^j}\binom{m}{j} \Big(\frac{me^2}{(k+1)!}\Big)^j\\
    &\leq \sum_{j=1}^m \frac{1}{j!} \Big(\frac{me^2}{(k+1)!}\Big)^j\\
    &\le \frac{(e-1)me^2}{(k+1)!}\,,
\end{align}
where the last line follows by using the inequality
\begin{equation}
    \Big(\frac{me^2}{(k+1)!}\Big)^j\leq \frac{me^2}{(k+1)}\,,
\end{equation}
since the r.h.s is less than unity. This proves the lemma.
\end{proof}
\begin{remark}
The proof of the above lemma goes through for any $j\leq m$ i.e., we have
\begin{equation}
    \|(l_0^j-l_0^{'j})l_0^{-j}\|\leq\frac{je^2(e-1)}{(k+1)!} \,,
\end{equation}
when $\|Ah\|\leq 1$ and $me^2/(k+1)!\leq 1$.
\end{remark}
We also need the following lemma.
\begin{lemma}\label{lem:L_1}
When $\|Ah\|\leq 1$ and $me^2/(k+1)!\leq 1$
\begin{equation}
    \|(L_1-L_1')L_0^{-1}\|\leq \frac{mTe^5}{(k+1)!}\,.
\end{equation}
\end{lemma}
\begin{proof}
First, we rewrite $L_1$ as follows.
\begin{equation}
    L_1=\int_0^T \exp(As) ds =\int_0^h\exp(As)ds + \int_h^{2h}\exp(As)ds +\dots +\int_{(m-1)h}^{mh}\exp(As) ds\,.
\end{equation}
This can be written as follows.
\begin{equation}
    L_1=\sum_{j=0}^{m-1}\int_0^h\exp(A(s+jh)) ds = \sum_{j=0}^{m-1}l_0^j\int_0^h\exp(As)ds\,.
\end{equation}
Finally, this can be written as
\begin{equation}
    L_1=\sum_{j=0}^{m-1}l_0^jl_1\,.
\end{equation}
Therefore
\begin{align}
    \|(L_1-L_1')L_0^{-1}\|&=\|\sum_{j=0}^{m-1}(l_0^jl_1-l_0^{'j}l_1')l_0^{-m}\|\\
    &\leq \sum_{j=0}^{m-1}\|(l_0^{j}-l_0^{'j})l_0^{-m}l_1'\| + \|l_0^{j-m}(l_1-l_1')\|\\
    &\leq \sum_{j=0}^{m-1}\Big[\|(l_0^{j}-l_0^{'j})l_0^{-m}l_1'\|+\|l_0^{-j}(l_1-l_1')\|\Big]\,.\label{eq:triangle}
\end{align}
Now we have
\begin{equation}
    \|l_1'\|=\|h\sum_{j=0}^k\frac{(Ah)^j}{(j+1)!}\|\leq he\,,
\end{equation}
and
\begin{equation}
    \|l_0^{'-j}\|\leq e\hspace{0.1in}\text{ and }\hspace{0.1in} \|l_0^{-j}\|\leq e\,.
\end{equation}
Using the above equations as well as \lem{L_0} and \lem{l_1} and plugging into \eq{triangle}, we get
\begin{equation}
    \|(L_1-L_1')L_0^{-1}\|\leq \sum_{j=0}^{m-1} \Big[\frac{je^2(e-1)}{(k+1)!}\frac{e^2}{\|A\|} + \frac{e^2}{\|A\|(k+1)!}\Big]\,.
\end{equation}
Therefore
\begin{equation}
    \|(L_1-L_1')L_0^{-1}\|\leq\frac{m(m-1)e^5}{\|A\|(k+1)!}=\frac{mTe^5}{(k+1)!}\,,
\end{equation}
where we used the fact that $m= T\|A\|+1$.
\end{proof}

\begin{lemma}\label{lem:AL_1}
We have the following relations.
\begin{equation}
    AL_1=L_0\hspace{0.1in}\text{and}\hspace{0.1in}AL_1'=L_0'\,.
\end{equation}
\end{lemma}
\begin{proof}
Consider $Al_1$. We have
\begin{equation}
    Al_1 = Ah\sum_{j=0}^\infty \frac{(Ah)^j}{(j+1)!} = \exp(Ah)-I=l_0-I\,.
\end{equation}
In a very similar way, we have $Al_1'=l_0'-I$. This means
\begin{equation}
    AL_1=A\sum_{j=0}^{m-1}l_0^jl_1=\sum_{j=0}^{m-1}l_0^j(l_0-I) = l_0^m-I = L_0-I\,.
\end{equation}
We can similarly show that
\begin{equation}
    AL_1'=L_0'-I\,.
\end{equation}
\end{proof}

We can now prove a bound on the solution error.
\begin{theorem}\label{thm:solution_error}
Suppose $x_T$ is the solution at time $T$ and $y_m$ is the solution given by the truncated Taylor series at $k$ terms, where
\begin{equation}
    (k+1)!\geq  \frac{me^3}{\delta}\Big(1+\frac{Te^2\|b\|}{\|x_T\|}\Big)\,,
\end{equation}
and $\|Ah\|\leq 1$, then we have the bound
\begin{equation}
    \|x_T-y_m\|\leq \delta\|x_T\|\,.
\end{equation}
\end{theorem}
\begin{proof}
From \lem{soln_form}, the solution after evolving for time $T$ is 
\begin{equation}
    x_T=L_0x_0+L_1b\,.
\end{equation}
Now the final state after $m$ time steps of the truncated Taylor series is 
\begin{equation}
    y_m=L_0'x_0 + L_1'b\,,
\end{equation}
where
\begin{equation}
    L'_0= T_k^m(Ah) \hspace{0.1in}\text{and}\hspace{0.1in} L_1'=h\sum_j T_k^j(Ah) S_k(Ah)\,.
\end{equation}
Now, let $x_h=L_0x_0$ be the homogeneous part of the solution, then the error is given by
\begin{align}
    \|x_T-y_m\|&=\Big\|(L_0-L'_0)x_0 + (L_1-L'_1)b\Big\|\\
    &=\Big\|(L_1-L'_1)L_0^{-1}Ax_h + (L_1-L'_1)b\Big\|\\
    &=\Big\|(L_1-L'_1)L_0^{-1}(Ax_h + L_0^{-1}b)\Big\|\\
    &=\Big\|(L_1-L'_1)L_0^{-1}(Ax_T + b)\Big\|\\
    &\leq \Big\|(L_0-L'_0)L_0^{-1}\Big\|\|x_T\|+\Big\|(L_1-L'_1)L_0^{-1}\Big\|\|b\|\\
    &\leq \frac{me^3}{(k+1)!}\|x_T\| + \frac{mTe^5}{(k+1)!}\|b\|\\
    &\leq \frac{me^3}{(k+1)!}\Big(1+\frac{Te^2\|b\|}{\|x_T\|}\Big)\|x_T\|\,,
\end{align}
where in the second line, we used \lem{AL_1}, the second-to-last line follows from the triangle inequality and the last line from \lem{L_0} and \lem{L_1}. Now, using the value of $k$ defined in the statement of this theorem, we have 
\begin{equation}
    \|x_T-y_m\|\leq \delta \|x_T\|\,.
\end{equation}
\end{proof}
\begin{remark}
Note that the value of $k$ in the above theorem does not depend on the condition number of the matrix $A$ or its diagonalizing matrix (if it is diagonalizable).
\end{remark}

\subsection{Condition number}
In this subsection, we prove a theorem that gives a bound on the condition number.
\begin{lemma}\label{lem:R_k_bound}
Recall that the error in the truncated Taylor series $T_k(At)$ of $\exp(At)$ is
\begin{equation}
    R_k(Ah)=\sum_{j=k+1}^\infty \frac{(Ah)^j}{j!}\,.
\end{equation}
Then for $\|Ah\|\leq 1$, the norm of this error satisfies
\begin{equation}
    \|R_k(Ah)\|\leq \frac{e}{(k+1)!}
\end{equation}
\end{lemma}
\begin{proof}
We can write the error norm as
\begin{equation}
    \|R_k(Ah)\|\leq \frac{(\|Ah\|)^{k+1}}{(k+1)!}\sum_{j=0}^\infty \frac{(\|Ah\|)^j}{j!}\frac{1}{\binom{k+j+1}{k}}\leq \frac{e}{(k+1)!}\,.
\end{equation}
\end{proof}
Next we prove a lemma on the norm of powers of the truncated Taylor series.
\begin{lemma}\label{lem:T_k_bound}
For $\|Ah\|\leq 1$ and $k$ as in \thm{solution_error}, we have that the truncated Taylor series $T_k(Ah)$ of $\exp(Ah)$ satisfies
\begin{equation}
    \|(T_k(Ah))^\ell\|\leq C(A)(1 + \delta)\,,
\end{equation}
for any $\ell\leq m$. Here $C(A)$ is from \defn{C(A)}.
\end{lemma}
\begin{proof}
We have
\begin{equation}
    \|(T_k(Ah))^\ell\|\leq \|\exp(Ah\ell)\| + \|(T_k(Ah))^\ell-\exp(Ah\ell)\|\leq \|\exp(Ah\ell)\|\Big(1 + \|(L_0-L'_0)L_0^{-1}\|\Big)\,.
\end{equation}
Using \lem{L_0} and \defn{C(A)}, this can be written as
\begin{equation}
    \|(T_k(Ah))^\ell\|\leq C(A)(1 + \delta)\,.
\end{equation}
\end{proof}

\begin{theorem}\label{thm:condition_number}
For any matrix $A$ and $k$ as in \thm{solution_error}, the condition number $\kappa_L$ of the linear system $L$ defined in \eq{linear_system} is bounded as 
\begin{equation}
    \kappa_L\leq (m+p)C(A)(1+\delta)e(1+e)\,.
\end{equation}
\end{theorem}

\begin{proof}
The condition number of $L$ is $\kappa_L=\|L\|\|L^{-1}\|$. We bound each of these below. The norm of $L$ is
\begin{equation}
    \|L\|=\|I-N\|\leq 1+\max\{1, \|M_2(I-M_1)^{-1}\|\}\,.
\end{equation}
To bound the norm of $\|M_2(I-M_1)^{-1}\|$, we can use its action on a state of the form $\ket{\ell,v}$ for some normalized basis state $\ket{v}$. From \eq{M_1M_2action}, this is
\begin{equation}
    M_2(I-M_1)^{-1}\ket{\ell,v}=\ket{0}\otimes T_{\ell,k}(Ah) \ket{v}\,,
\end{equation}
where
\begin{equation}
    T_{\ell,k}(Ah)=\sum_{j=0}^{k-\ell}\frac{\ell!(Ah)^j}{(\ell+j)!}\,.
\end{equation}
We can see that for all $0\leq\ell\leq k$ (since $\|Ah\|\leq 1$), we have the following bound.
\begin{equation}\label{eq:T_ell_k}
    \|T_{\ell,k}(Ah)\|\leq \sum_{j=0}^{k-\ell}\frac{1}{j!\binom{\ell+j}{j}}\leq \sum_{j=0}^{k-\ell}\frac{1}{j!}\leq e\,.
\end{equation}
Now for an arbitrary state $\ket{\tilde{v}}$ of the form
\begin{equation}
    \ket{\tilde{v}} = \sum_{\ell} a_{\ell}\ket{\ell,v_\ell}\,,
\end{equation}
where $\sum_{\ell}|a_{\ell}|^2=1$.
Then we have
\begin{equation}
    \|M_2(I-M_1)^{-1}\ket{\tilde{v}}\|\leq \sum_{\ell}|a_{\ell}|\|M_2(I-M_1)^{-1}\ket{\ell,v_\ell}\|\leq \sum_\ell |a_\ell| e\,.
\end{equation}
Since $\sum_\ell |a_\ell|^2=1$, we have 
\begin{equation}
    \sum_\ell |a_\ell|\leq \sqrt{k}\,.
\end{equation}
Using this, we have
\begin{equation}
    \|M_2(I-M_1)^{-1}\| \leq \sqrt{k}e\,.
\end{equation}
Now to bound $\|N\|$, we look at the action of powers of $N$ on a vector of the form $\ket{i,\ell,v}$. Using \eq{M_1M_2action}, we get
\begin{equation}
    N^q\ket{i,\ell,v}=T_k^{q-1}(Ah)T_{\ell,k}(Ah)\,.
\end{equation}

Now using \lem{T_k_bound} and \eq{T_ell_k}, we have
\begin{equation}
    \|N^q\|\leq C(A)(1+\delta)e\,.
\end{equation}
This gives us
\begin{equation}\label{eq:L_inv}
    \|L^{-1}\|=\|(I-N)^{-1}\|\leq (m+p)C(A)(1+\delta)e\,.
\end{equation}
Finally, the condition number can be bounded as
\begin{equation}
    \kappa_L\leq O((m+p)C(A)\sqrt{k}(1+\delta))\,.
\end{equation}
\end{proof}

\subsection{Probability of success}
\begin{theorem}\label{thm:prob_success}
For $k$ as in \thm{solution_error} and $\delta\leq 1/2$, the probability of success of obtaining one of the outcomes in the set $\{m,m+1,\dots m+p-1\}$ when measuring the time register is lower bounded as follows. When $m=p$, we have
\begin{equation}
    P_{meas}\geq \frac{1}{18g^2}\,,
\end{equation}
where
\begin{equation}
    g=\frac{\max_{t\in [0,T]}\|x(t)\|}{\|x_T\|}\,.
\end{equation}
\end{theorem}
\begin{proof}
First, let $\ket{y_g}$ be defined as follows.
\be
\ket{y_g}=\sum_{k=m}^{m+p-1}\ket{k}\ket{y_k}\,.
\ee
The probability of a successful measurement is
\be
P_{meas}=\frac{\| \ket{y_g}\|^2}{\|\ket{y}\|^2}\,,
\ee
where $y$ is the state from \eq{final_state}.

The squared norm of $y$ is
\begin{equation}
    \|y\|^2=\sum_{i=0}^m\|y_i\|^2 + p\|y_m\|^2\leq (m+p)\max_{i}\|y_i\|^2\,.
\end{equation}
Let $i_0$ be the index that maximizes $\|y_i\|^2$. Then, using \thm{solution_error}, we have
\begin{align}
    \|y_{i_0}\|&\leq \|x(i_0h)\|+\frac{me^3}{(k+1)!}\|x(i_0h)\|+\frac{mTe^5}{(k+1)!}\|b\|\\
    &\leq \max_{t\in [0,T]}\|x(t)\|+\frac{me^3}{(k+1)!}\max_{t\in [0,T]}\|x(t)\| + \frac{mTe^5}{(k+1)!}\|b\|\\
    &\leq \max_{t\in [0,T]}\|x(t)\|+\frac{me^3}{(k+1)!}\Big(1+\frac{Te^2\|b\|}{\max_{t\in [0,T]}\|x(t)\|}\Big)\max_{t\in [0,T]}\|x(t)\|\\
    &\leq \max_{t\in [0,T]}\|x(t)\|+\frac{me^3}{(k+1)!}\Big(1+\frac{Te^2\|b\|}{\|x_T\|}\Big)\max_{t\in [0,T]}\|x(t)\|\\
    &\leq(1+\delta)\max_{t\in [0,T]}\|x(t)\|\,.
\end{align}
Therefore, we have
\begin{equation}
    P_{meas}\geq\frac{\| \ket{y_g}\|^2}{(m+p)(1+\delta)^2\max_{t\in [0,T]}\|x(t)\|^2}\,,
\end{equation}
Since the solution error is $\|x_T-y_m\|\leq \delta\|x_T\|$ from \thm{solution_error}, we have
\begin{align}
    P_{meas}&\geq \frac{p\|x_T + y_m-x_T\|^2}{(m+p)(1+\delta)^2\max_{t\in [0,T]}\|x(t)\|^2}\\
    &\geq \frac{p(\|x_T\| - \|x_T-y_m\|)^2}{(m+p)(1+\delta)^2\max_{t\in [0,T]}\|x(t)\|^2}\\
    &\geq \frac{p(1-\delta)^2}{(m+p)(1+\delta)^2g^2}\,.
\end{align}
where
\begin{equation}\label{eq:g}
    g=\frac{\max_{t\in [0,T]}\|x(t)\|}{\|x(T)\|}\,.
\end{equation}
Choosing $m=p$ and $\delta\leq 1/2$, we get 
\be
P_{meas}\geq \frac{1}{18g^2}\,.
\ee
\end{proof}

\subsection{Initial state preparation and circuit implementation}
The preparation scheme for the initial state is very similar to the one in \cite{BCOW17}.
\begin{lemma}\label{lem:initial_state}
Let $O_x$ and $O_b$ be the oracles from \defn{oracles}. Then the state proportional to $\ket{\psi_{in}}$ can be produced with a single call to $O_x$ and $O_b$ and an additional $\polylog(m)$ elementary gates.
\end{lemma}
\begin{proof}
Recall that the initial state is
\begin{equation}
    \ket{\psi_{in}}=\frac{1}{N_{init}}(\ket{0,0,x_0} + h\sum_{i=0}^{m-1}\ket{i,1,b})\,,
\end{equation}
where
\begin{equation}
    N_{init}=\sqrt{\|x_0\|^2 + mh^2\|b\|^2}\,.
\end{equation}
To prepare this state from $\ket{0,0,0}$, first apply the following rotation on the second register.
\begin{equation}
    \ket{0,0,0}\rightarrow \frac{\|x_0\|}{N_{init}}\ket{0,0,0} + \frac{\sqrt{m}h\|b\|}{N_{init}}\ket{0,1,0}\,.
\end{equation}
This step takes $O(1)$ gates since it is a unitary on a single qubit. Next, we apply the oracles $O_x$ and $O_b$ to the third register conditioned on the second register being in $\ket{0}$ and $\ket{1}$ respectively to get
\begin{equation}
    \frac{\|x_0\|}{N_{init}}\ket{0,0,\bar{x}_0} + \frac{\sqrt{m}h\|b\|}{N_{init}}\ket{0,1,\bar{b}}=\frac{1}{N_{init}}\ket{0,0,x_0} + \frac{\sqrt{m}h}{N_{init}}\ket{0,1,b}\,.
\end{equation}
Finally, conditioned on the second register being $\ket{1}$, we implement the following rotation on the first register that takes
\begin{equation}
    \ket{0}\rightarrow \frac{1}{\sqrt{m}}\sum_{i=0}^{m-1}\ket{i}\,.
\end{equation}
This step takes $\polylog(m)$ gates. All this gives us the normalized initial state we want with the complexity stated in the theorem.
\end{proof}

We now compute the query and gate complexity of implementing the QLSA on the linear system \eq{linear_system}.
\begin{theorem}\label{thm:gate_complexity}
The QLSA algorithm applied to the linear system \eq{linear_system} has a query complexity of
\begin{equation}
    O(s k\kappa_L\polylog (k,m,d,\kappa_L,1/\epsilon))\,,
\end{equation}
and gate complexity greater by a factor of at most 
\begin{equation}
    O(\polylog(k,m,1/\epsilon))\,,
\end{equation}
where $s$ is the sparsity of the matrix $A$, $d$ is its dimension, $\kappa_L$ is the condition number of the linear system, $m$ is the number of time steps of the algorithm and $k$ is the number of terms before truncation of the Taylor series.
\end{theorem}
\begin{proof}
It can be seen from the structure the matrix $M_1$, that it is an $s_r$ row sparse and $s_c$ column sparse matrix since $A$ is an $(s_r,s_c)$ sparse matrix. $M_1$ is of dimension $kd$, where $d$ is the dimension of $A$. This means that $I-M_1$ is a $(s_r+1,s_c+1)$-sparse matrix. The condition number of $I-M_1$ can be computed as follows. First, the norm of $I-M_1$ is
\begin{equation}
    \|I-M_1\|\leq 2\,,
\end{equation}
since $\|M_1\|\leq 1$. Now for the norm of its inverse, we have
\begin{equation}
    \|(I-M_1)^{-1}\|= \|I+M_1+\dots +M_1^{k-1}\|\leq k\,,
\end{equation}
Therefore, the condition number $\kappa$ of $I-M_1$ is
\begin{equation}
    \kappa\leq 2k\,.
\end{equation}

We now bound the query complexity in the following steps.
\begin{enumerate}
    \item Using \lem{sparse_matrix}, we can block encode the matrix $I-M_1$ as a $(\sqrt{s_rs_c}, O(\log(d k)), \epsilon_0)$ block unitary with complexity 
    \begin{equation}
        T_0 = O(\log(dk) + \log^{2.5}(s/\epsilon_0))\,.
    \end{equation}
    This is because $I-M_1$ has dimension $dk$ and hence is a $O(\log dk )$ qubit matrix.
    
    \item Using \lem{matrix_inversion}, the inverse of $(I-M_1)^{-1}$ can be implemented as an $(\alpha_1,a_1,\epsilon_1)$ block unitary with complexity
    \begin{equation}
        T_1=O(s k(\log(d k)+\log^{2.5}(k/\epsilon_1)))\,,
    \end{equation} 
    where $s=\max\{s_r,s_c\}$, $\alpha_1=2k$ and $a_1=\log(d k)+O(\log(2k/\epsilon_1))$. The fact that $\alpha_1=2k$ comes from the condition number of $I-M_1$.
    
    \item Using \lem{sparse_matrix}, since the matrix $M_2$ is a $(k,1)$-sparse matrix with nonzero entries equal to $1$, $M_2$ can be $(\alpha_2,a_2,\epsilon_2)$ block encoded with complexity
    \begin{equation}
        T_2=O(\log(d k)+\log^{2.5}(k/\epsilon_2))\,,
    \end{equation}
    where $\alpha_2=\sqrt{k}$ and $a_2=\log(d k)+3$.
    
    \item Using \lem{matrix_arithmetics}, the matrix $M_2(I-M_1)^{-1}$ can be $(\alpha_3,a_3,\epsilon_3)$ block encoded with complexity
    \begin{equation}
        T_3=O(T_1+T_2)\,,
    \end{equation}
    where $\alpha_3=\alpha_1\alpha_2$, $a_3=a_1+a_2$ and $\epsilon_3=\alpha_1\epsilon_2+\alpha_2\epsilon_1$.
    
    \item Using \lem{matrix_arithmetics} again, we can implement $I-N$ as a $(\alpha_3,a,\epsilon_3)$ block matrix with complexity $T_4=O(T_3)$, where $a=O(\log(kdm/\epsilon_1))$.
    
    \item Finally, using \thm{block_QLSA}, we apply the block QLSA algorithm with complexity
    \begin{equation}\label{eq:T_QLSA}
        O(\kappa_L \log(1/\epsilon)) (T_4+T_\psi)\,,
    \end{equation}
    where $T_\psi$ is the complexity of preparing the initial state.
\end{enumerate}
For the total error of the implementation to be bounded by $\epsilon$, we first choose $\epsilon_1=\epsilon_2$. We then need
\begin{equation}
    \epsilon_3= \sqrt{k}\epsilon_1 + 2k\epsilon_1\leq O(\frac{\epsilon}{\kappa_L\log 1/\epsilon})\,.
\end{equation}
For this to be satisfied, we need 
\begin{equation}
    \epsilon_1\leq O(\frac{\epsilon}{\kappa_L k\log 1/\epsilon})\,.
\end{equation}
Plugging this into expressions for $\alpha_3$, $T_4$ and $a$, we get
\begin{align}
    &\alpha_3=O(k^{1.5})\\
    &a=O(\log(m\kappa_L d k\log(1/\epsilon)/\epsilon))\\
    &T_4=O( (sk(\log(d k) + \log^{2.5}(k\kappa_L\log(1/\epsilon)/\epsilon)))\,.
\end{align}
Taking $T_\psi$ from \lem{initial_state} and plugging all this into \eq{T_QLSA}, we get the query and gate complexities in the statement of the theorem.
\end{proof}

\section{Main result}\label{sec:main}
We now state and prove our main result.
\begin{theorem}\label{thm:main_thm}
For any sparse matrix $A$ with sparsity $s$, dimension $d$ and $C(A)$ from \defn{C(A)} and for the problem of producing a quantum state proportional to the solution of a linear nonhomogeneous differential equation defined in \prb{ODE}, there exists a quantum algorithm that produces a quantum state which is $\epsilon$ close to the normalized solution with overall query complexity given by
\begin{equation}
    O\Big(gT\|A\|C(A)\poly\Big(s,\log d,\log(1+\frac{Te^2\|b\|}{\|x_T\|}),\log(\frac{1}{\epsilon}),\log(T\|A\|C(A))\Big)\Big)\,,
\end{equation}
and gate complexity which is greater by a factor of at most
\begin{equation}
    O(\polylog (1+\frac{Te^2\|b\|}{\|x_T\|}, 1/\epsilon, T\|A\|))\,,
\end{equation}
where from \eq{g}, $g$ is defined as
\begin{equation}
    g=\frac{\max_{t\in [0,T]}\|x(t)\|}{\|x(T)\|}\,.
\end{equation}
\end{theorem}
\begin{proof}
We pick the parameters of the algorithm as follows.
\begin{align}\label{eq:parameters}
    &h= \frac{T}{\lceil T\|A\|\rceil}\,,m=p=\frac{T}{h}=\lceil T\|A\|\rceil\,,\delta\leq \frac{\epsilon}{2}\\
    &k=\Big\lceil\frac{2\log\Omega}{\log\log\Omega}\Big\rceil\,,
\end{align}
where 
\begin{equation}
    \Omega=e^3T\|A\|\Big(1+\frac{Te^2\|b\|}{\|x_T\|}\Big)\,.
\end{equation}
This choice of $k$ makes $(k+1)!>\Omega$. Let the output of the algorithm described in \sec{algorithm} be the state $\ket{\psi}$. In \sec{algorithm}, we showed that the quantum state $\ket{\psi}$ produced by the algorithm is $y_m/\|y_m\|$. In \thm{solution_error}, it was shown that
\begin{equation}
    \|y_m-x_T\|\leq \delta \|x_T\|\,.
\end{equation}
This means, we have
\begin{align}
    &\Big\|\frac{y_m}{\|y_m\|}-\frac{x_T}{\|x_T\|}\Big\|\leq\Big\|\frac{y_m}{\|x_T\|}-\frac{x_T}{\|x_T\|}\Big\|+\Big\|\frac{y_m}{\|y_m\|}-\frac{y_m}{\|x_T\|}\Big\| \\
    &\leq \delta + \frac{\Big|\|x_T\|-\|y_m\|\Big|}{\|x_T\|}\\
    &\leq \delta + \frac{\|x_T-y_m\|}{\|x_T\|}\\
    &\leq 2\delta
    \leq \epsilon\,.
\end{align}
From \thm{prob_success}, the success probability is proportional to $O(1/g^2)$. To make the success probability a constant, we need to perform amplitude amplification \cite{Amplitude_amp}. Taking this into account the overall complexity increases by a factor of $g$. Using the bound on the condition number from \thm{condition_number} and bounds on query and gate complexities from \thm{gate_complexity}, after plugging in the values from \eq{parameters}, we get the query and gate complexities of the algorithm as in the theorem statement.
\end{proof}

\section{Application to nonlinear differential equations}\label{sec:nonlinear}
Consider the nonlinear differential equation from \cite{Liue2026805118}, which is of the form
\begin{equation}\label{eq:quadratic_ode}
    \frac{du}{dt}=F_2u^{\otimes 2}+F_1u + F_0\,,\hspace{0.5in}u(0)=u_{in}\,.
\end{equation}
where $u\in \mathbb{R}^d$, $F_2\in \mathbb{R}^{d\times d^2}$, $F_1\in \mathbb{R}^{d\times d}$ and $F_0\in\mathbb{R}^d$. Unlike in \cite{Liue2026805118}, we assume that $F_0$ (and $F_1$ and $F_2$) are time-independent\footnote{$F_0$ can be assumed to be time-dependent like in \cite{Liue2026805118}. Here we leave out the time-dependence for simplicity.}. However, we do not assume that $F_1$ is diagonalizable. We only assume that $F_1$ has a negative log-norm $\mu(F_1)< 0$. 

The procedure of Carleman linearization takes a quadratic (or higher order) ODE and converts it into a linear equation of higher dimension. For the quadratic ODE above, we get
\begin{equation}
  \frac{\d{x}}{\d{t}} = A  x + b, \qquad
   x(0) = x_{\mathrm{in}}
\label{eq:LODE}
\end{equation}
with the tri-diagonal block structure
\begin{equation}
\frac{\d{}}{\d{t}}
  \begin{pmatrix}
     x_1 \\
     x_2 \\
     x_3 \\
    \vdots \\
     x_{N-1} \\
     x_N \\
  \end{pmatrix}
=
  \begin{pmatrix}
    A_1^1 & A_2^1 &  &  &  &  \\
    A_1^2 & A_2^2 & A_3^2 & &  &  \\
     & A_2^3 & A_3^3 & A_4^3 &  &  \\
     &  & \ddots & \ddots & \ddots &  \\
     &  &  & A_{N-2}^{N-1} & A_{N-1}^{N-1} & A_N^{N-1} \\
     &  &  &  & A_{N-1}^N & A_N^N \\
  \end{pmatrix}
  \begin{pmatrix}
     x_1 \\
     x_2 \\
     x_3 \\
    \vdots \\
     x_{N-1} \\
     x_N \\
  \end{pmatrix}+
  \begin{pmatrix}
    F_0 \\
    0 \\
    0 \\
    \vdots \\
    0 \\
    0 \\
  \end{pmatrix},
\label{eq:UODE}
\end{equation}
where $x_j=u^{\otimes j}\in\R^{d^j}$, $x_{\mathrm{in}}=[u_{\mathrm{in}}; u_{\mathrm{in}}^{\otimes 2}; \ldots; u_{\mathrm{in}}^{\otimes N}]$, and $A_{j+1}^j \in \R^{d^j\times d^{j+1}}$, $A_j^j \in \R^{d^j\times d^j}$, $A_{j-1}^j \in \R^{d^j\times d^{j-1}}$ for $j\in\range{N}$ satisfying
\begin{align}
A_{j+1}^j &= F_2\otimes I^{\otimes j-1}+I\otimes F_2\otimes I^{\otimes j-2}+\cdots+I^{\otimes j-1}\otimes F_2, \label{eq:tensor2} \\
A_j^j &= F_1\otimes I^{\otimes j-1}+I\otimes F_1\otimes I^{\otimes j-2}+\cdots+I^{\otimes j-1}\otimes F_1, \label{eq:tensor1} \\
A_{j-1}^j &= F_0\otimes I^{\otimes j-1}+I\otimes F_0\otimes I^{\otimes j-2}+\cdots+I^{\otimes j-1}\otimes F_0. \label{eq:tensor0}
\end{align}
Note that $A$ is a $(3Ns)$-sparse matrix, where $s$ is the sparsity of $A$. The dimension of \eq{LODE} is
\begin{equation}
  \Delta \coloneqq d+d^2+\cdots+d^N=\frac{d^{N+1}-d}{d-1}=O(d^N).
\end{equation}

We now define a quantity similar to the one that was used in \cite{Liue2026805118} to quantify the amount of nonlinearity in the quadratic ODE \eq{quadratic_ode}. The difference is that we use log-norm here instead of $\max_\lambda\Re(\lambda)$.
\begin{definition}
\begin{equation}\label{eq:R}
    R=\frac{1}{|\mu(F_1)|}\Big(\|F_2\|\|u_{in}\| + \frac{\|F_0\|}{\|u_{in}\|}\Big)\,.
\end{equation}
\end{definition}
The following lemmas are needed to give a bound on the error in truncating the linearization procedure at level $N$.
\begin{restatable}{lemma}{decrease}\label{lem:Carleman}
For the quadratic ODE in \eq{quadratic_ode}, when $R<1$ we have
\begin{equation}
    \|u(t)\|\leq \|u_{in}\|\,.
\end{equation}
\end{restatable}
\begin{proof}
The proof of this lemma is given in \appx{appendix_A}.
\end{proof}
Define the error due to truncation of Carleman procedure as follows.
\begin{equation}
    \eta_j(t)=u^{\otimes j}(t) - x_j(t)\,,
\end{equation}
where $\eta_1(t)$ is the error between the solution of \eq{quadratic_ode} and $x_1(t)$ (which will be $\epsilon$ close to the output of the quantum algorithm). Let $\eta(t)$ be the vector whose entries are $\eta_j$ for $j=1$ to $N$.
\begin{lemma}\label{lem:C(A)_nlode}
For the quadratic ODE from \eq{quadratic_ode}, let the linear ODE obtained by Carleman linearization be 
    \begin{equation}
        \frac{dx}{dt}=Ax+b\,.
    \end{equation}
Suppose $|\mu(F_1)|>\|F_0\|+\|F_2\|$ and $\mu(F_1)<0$, then
    \begin{equation}
        C(A)\leq N\,.
    \end{equation}
\end{lemma}
\begin{proof}
Let $D$ be a block-diagonal matrix with diagonal blocks given by $I/j$ (which makes it also a diagonal matrix). Here each identity matrix of size $\R^{d^j\times d^j}$ for the $j^{th}$ block with a total of $N$ blocks. Now consider the matrix $Q=DA+A^\dag D$. We are interested in the log-norm of $Q$. First split $A=H_0+H_1+H_2$, where
\begin{align}
    H_0 &= \sum_{j=2}^N \ket{j}\bra{j-1}\otimes A^j_{j-1}\\
    H_1 &= \sum_{j=1}^N \ket{j}\bra{j}\otimes A^j_j\\
    H_2 &= \sum_{j=1}^{N-1} \ket{j}\bra{j+1}\otimes A^j_{j+1}\,.
\end{align}
Then $Q$ can be split as
\begin{equation}
    Q=(DH_0 + H_0^\dag D) + (DH_1 + H_1^\dag D) + (DH_2 + H_2^\dag D) 
\end{equation}
We have
\begin{equation}
    \mu(Q)= \frac{1}{2}\sup_{x:\|x\|=1}\bra{x}(DH_1+H_1^\dag D)\ket{x} + \frac{1}{2}\sup_{x:\|x\|=1}\bra{x}(D(H_0+H_2) + (H_0^\dag + H_2^\dag)D)\ket{x}\,.
\end{equation}
Now
\begin{equation}
    \mu(DH_1) = \mu(F_1)\,,
\end{equation}
and 
\begin{align}
    \mu(D(H_0+H_2))&\leq \frac{1}{2}(\|D(H_0+H_2)\| + \|(H_0^\dag + H_2^\dag)D\|)\\
    &\leq \|F_0\|+\|F_2\|\,.
\end{align}
This gives us
\begin{equation}
    \mu(Q)\leq \mu(F_1) + \|F_0\|+\|F_2\|\,.
\end{equation}
The above quantity is negative based on our assumption that $R<1$.

Next, we use a result from \cite{Plischke_2005} (see also \cite{Jennings_2024}, where this was first used in quantum algorithms for differential equations). Using this, we get
\begin{equation}
    \|\exp(At)\|\leq \kappa(D)=N\,.
\end{equation}
This gives us
\begin{equation}
    C(A)\leq N\,.
\end{equation}
\end{proof}

Before we state and prove the next lemma, we give a definition of rescaling that we use here\footnote{Note that rescaling was already introduced in \cite{Liue2026805118}.}.
Consider the quadratic ODE from \eq{quadratic_ode} defined as follows.
\begin{equation}\label{eq:q_ode}
    \frac{du}{dt} = F_0 + F_1u + F_2u^{\otimes 2}\,,\hspace{0.2in} u(0)=u_{in}\,.
\end{equation}
It has been shown in \cite{Liue2026805118} that this can be rescaled by a factor $\gamma$ as follows. Define a new variable $y$ such that $u=\gamma y$. In terms of $y$, the above equation becomes
\begin{equation}
    \frac{dy}{dt} = \frac{F_0}{\gamma} + F_1y + \gamma F_2 y^{\otimes 2}\,,\hspace{0.2in} y(0) = \frac{u_{in}}{\gamma}\,.
\end{equation}
Using this kind of rescaling, we can make the norm of the initial state strictly less than unity (as was done in \cite{Liue2026805118}). We are now ready to state the next lemma.
\begin{lemma}\label{lem:A_N}
Suppose $R<1$ and suppose there is a rescaling such that we have (after rescaling)
\begin{enumerate}
    \item $C(A)\leq N$ and,
    \item $\|u(0)\|<1$.
\end{enumerate}
Then, if we choose 
\begin{equation}
    N\geq \Bigg\lceil\frac{2\log(T\|F_2\|/\delta \|u(T)\|)}{\log(1/\|u(0)\|)}\Bigg\rceil\,,
\end{equation}
we have that the norm of the error $\|\eta_1(t)\| = \|x_1(t) -u(t)\|$ coming from truncating Carleman linearization to $N$ steps can be bounded as 
\begin{equation}
    \|\eta_1(T)\|\leq \delta\|u(T)\|\,.
\end{equation}
\end{lemma}
\begin{remark}
When $R<1$, we show in \lem{conditions}, the existence of a rescaling that satisfies the above conditions (1) and (2).    
\end{remark}
\begin{proof}
First, note that the exact solution $u(t)$ of the original quadratic ODE \eq{quadratic_ode} satisfies
\begin{equation}
\frac{\d{}}{\d{t}}
  \begin{pmatrix}
    u \\
    u^{\otimes 2} \\
    u^{\otimes 3} \\
    \vdots \\
    u^{\otimes(N-1)} \\
    u^{\otimes N} \\
    \vdots \\
  \end{pmatrix}
=
  \begin{pmatrix}
    A_1^1 & A_2^1 &  &  &  &  &  \\
    A_1^2 & A_2^2 & A_3^2 & &  &  &  \\
     & A_2^3 & A_3^3 & A_4^3 &  &  &  \\
     &  & \ddots & \ddots & \ddots &  &  \\
     &  &  & A_{N-2}^{N-1} & A_{N-1}^{N-1} & A_N^{N-1} &  \\
     &  &  &  & A_{N-1}^N & A_N^N & \ddots \\
     &  &  &  &  & \ddots & \ddots \\
  \end{pmatrix}
  \begin{pmatrix}
    u \\
    u^{\otimes 2} \\
    u^{\otimes 3} \\
    \vdots \\
    u^{\otimes(N-1)} \\
    u^{\otimes N} \\
    \vdots \\
  \end{pmatrix}+
  \begin{pmatrix}
    F_0 \\
    0 \\
    0 \\
    \vdots \\
    0 \\
    0 \\
    \vdots \\
  \end{pmatrix},
\label{eq:LUODE}
\end{equation}
and the approximated solution satisfies \eq{UODE}. Comparing these equations, we have the tri-diagonal block structure
\begin{equation}
\frac{\d{}}{\d{t}}
  \begin{pmatrix}
    \eta_1 \\
    \eta_2 \\
    \eta_3 \\
    \vdots \\
    \eta_{N-1} \\
    \eta_N \\
  \end{pmatrix}
=
  \begin{pmatrix}
    A_1^1 & A_2^1 &  &  &  &  \\
    A_1^2 & A_2^2 & A_3^2 & &  &  \\
     & A_2^3 & A_3^3 & A_4^3 &  &  \\
     &  & \ddots & \ddots & \ddots &  \\
     &  &  & A_{N-2}^{N-1} & A_{N-1}^{N-1} & A_N^{N-1} \\
     &  &  &  & A_{N-1}^N & A_N^N \\
  \end{pmatrix}
  \begin{pmatrix}
    \eta_1 \\
    \eta_2 \\
    \eta_3 \\
    \vdots \\
    \eta_{N-1} \\
    \eta_N \\
  \end{pmatrix}+
  \begin{pmatrix}
    0 \\
    0 \\
    0 \\
    \vdots \\
    0 \\
    A_{N+1}^Nu^{\otimes(N+1)} \\
  \end{pmatrix},
\label{eq:EUODE}
\end{equation}
which we write compactly as
\begin{equation}
  \frac{\d{\eta}}{\d{t}} = A \eta + \hat b(t), \qquad
  \eta(0) = 0
\label{eq:ELODE}
\end{equation}
This can be integrated to give
\begin{equation}
    \eta(t) = \int_{0}^t e^{A(t-s)}\hat{b}(s) ds\,.
\end{equation}
The norm of the error satisfies
\begin{equation}
    \|\eta\|\leq \int_{0}^t\|e^{A(t-s)}\|\|\hat{b}(s)\|ds\,.
\end{equation}
When $R<1$ and $\mu(F_1)<0$, using \lem{Carleman}, we have $\|u\|\leq \|u(0)\|$ and using \lem{C(A)_nlode}, we have
\begin{equation}
    \|e^{A(t-s)}\|\leq C(A)\leq N\,.
\end{equation}
This means we have
\begin{equation}
    \|\hat{b}(s)\|\leq \|A^N_{N+1}\|\|u^{N+1}\|\leq N\|F_2\|\|u(0)\|^{N+1}\,.
\end{equation}
We can re-write the error as follows.
\begin{equation}
    \|\eta\|\leq N^2\|F_2\|T\|u(0)\|^{N+1}\,.
\end{equation}
Since we have rescaled to pick $\|u(0)\|<1$, we have that $\|\eta\|$ converges to zero. Therefore, if we choose
\begin{equation}
    N\geq \Bigg\lceil\frac{2\log(T\|F_2\|/\delta \|u(T)\|)}{\log(1/\|u(0)\|)}\Bigg\rceil\,,
\end{equation}
we get
\begin{equation}
    \|\eta\|\leq \delta \|u(T)\|\,.
\end{equation}
\end{proof}

Finally, we need another result (whose proof is in \appx{appendix_B}).
\begin{restatable}{lemma}{conds}\label{lem:conditions}
For the quadratic ODE \eq{quadratic_ode}, suppose $R<1$. There exists a rescaling $\gamma$ i.e., $u=\gamma y$ such that we have
\begin{enumerate}
    \item $C(\tilde{A})\leq N$ and,
    \item $\|y(0)\|<1$,
\end{enumerate}
 where $\tilde{A}$ denotes the new $A$ matrix after rescaling (comprised of the new $F$ matrices denoted $\tilde{F}_i$).
\end{restatable}

We are now ready to prove the main result in this section.
\begin{theorem}
There is an efficient quantum algorithm to solve the quadratic ODE from \eq{quadratic_ode} i.e., to produce a quantum state $\epsilon$ close to a state proportional to the solution with query complexity
\begin{equation}
    O\Big(g_uT\|A\|\poly\Big(N,s,\log(1+\frac{\|F_0\|}{\|u(T)\|}),\log(\frac{1}{\epsilon}),\log(T\|A\|)\Big)\Big)\,,
\end{equation}
and gate complexity larger by a factor of at most
\begin{equation}
    O\Big(\polylog \Big(n,(1+\frac{\|F_0\|}{\|u(T)\|}),\frac{1}{\epsilon},T\|A\|\Big)\Big)\,,
\end{equation}
where 
\begin{align}
    N &= \Bigg\lceil\frac{2\log(T\|F_2\|/\delta \|u(T)\|)}{\log(1/\|u_{in}\|)}\Bigg\rceil\label{eq:N_def}\\
    g_u&=\frac{\|u_{in}\|}{\|u(T)\|}\,.
\end{align}
\end{theorem}
\begin{proof}
After rescaling to make the norm of the initial state smaller than unity, we apply the truncated Carleman linearization with $N$ as above in \eq{N_def}.
\begin{equation}
    N=\Bigg\lceil\frac{2\log(T\|F_2\|/\delta \|u(T)\|)}{\log(1/\|u_{in}\|)}\Bigg\rceil\,.
\end{equation}
We get a linear ODE
\begin{equation}
    \frac{\d{x}}{\d{t}}=Ax + b\,,
\end{equation}
with the following parameters.
\begin{enumerate}
    \item Sparsity of $A$ is $3Ns$, where $s$ is the sparsity of $F_0$, $F_1$ and $F_2$.
    \item Dimension of $A$ is $D=O(d^N)$, where $d$ is the dimension of $F_1$.
    \item $C(A)\leq N$ as shown in \lem{C(A)_nlode}.
    \item A bound on $\|A\|$ can be found as follows. As before, split as $A=H_0+H_1+H_2$, where
    \begin{align}
        H_0 &= \sum_{j=2}^N \ket{j}\bra{j-1}\otimes A^j_{j-1}\\
        H_1 &= \sum_{j=1}^N \ket{j}\bra{j}\otimes A^j_j\\
        H_2 &= \sum_{j=1}^{N-1} \ket{j}\bra{j+1}\otimes A^j_{j+1}\,.
    \end{align}
    We can then write
    \begin{equation}
        \|A\|\leq \|H_0\|+\|H_1\|+\|H_2\|\leq N(\|F_0\| + \|F_1\| + \|F_2\|)\,.
    \end{equation}
    \item To bound the quantity $\|x(T)\|$, note that for all $t$
    
    \begin{equation}
        \begin{aligned}
        \|x(t)\|^2
        &= \sum_{j=1}^N\|x_j(t)\|^2
        \leq \sum_{j=1}^N(1+\delta)^2\|u^{\otimes j}(t)\|^2
        \le (1+\delta)^2N\|u(t)\|^2 \\
        \|x(t)\|^2 &= \sum_{j=1}^N\|x_j(t)\|^2 \ge \sum_{j=1}^N(1-\delta)^2\|u^{\otimes j}(t)\|^2\ge (1-\delta)^2\|u(t)\|^{2}\,.
        \label{eq:parallel_inequality}
        \end{aligned}
    \end{equation}
    \item Recall that the quantity $g$ is defined in \eq{g} as follows
    \begin{equation}
        g=\frac{\max_{t\in [0,T]}\|x(t)\|}{\|x(T)\|}\,.
    \end{equation}
    Using \eq{parallel_inequality} above and that $\max_{t\in [0,T]}\|u(t)\|=\|u_{in}\|$ when $R<1$, we get
    \begin{equation}
        g\leq \frac{1+\delta}{1-\delta}\sqrt{N}g_u\,,
    \end{equation}
    where (as above)
    \begin{equation}
        g_u=\frac{\|u_{in}\|}{\|u(T)\|}\,.
    \end{equation}
    When $\delta\leq 1/2$, we have
    \begin{equation}
        g\leq 3\sqrt{N}g_u\,.
    \end{equation}
\end{enumerate}
The initial state preparation is the same as in Lemma 5 of \cite{Liue2026805118}, where it is shown how to prepare the normalized version of
\begin{equation}
    \ket{z_{in}}\propto \sum_{j=1}^N\ket{j}\otimes \ket{u_{in}^{\otimes j}}\otimes \ket{0^{N-j}}\,,
\end{equation}
where $j$ labels the truncation level. This state is $x_0$ in the algorithm from \sec{algorithm} and is used to prepare the initial state for that algorithm.

From \thm{main_thm}, we can use the algorithm from \sec{algorithm} to get a quantum state $y$ such that 
\begin{equation}
    \|y-x(T)\|\leq \delta' \|x(T)\|\leq \delta'(1+\delta)\sqrt{N}\|u(T)\|\,.
\end{equation}
Therefore, the component of $y$ corresponding to the truncation level $1$ also satisfies
\begin{equation}
    \|y_1-x_1(T)\|\leq \delta'(1+\delta)\sqrt{N}\|u(T)\|\,.
\end{equation}
Now, measuring the quantum register containing the truncation level of Carleman linearization gives outcome $1$ with probability 
\begin{equation}
    \frac{\|y_1(T)\|^2}{\|y(T)\|^2}\geq \frac{(1-\delta)^2(1-\delta')^2\|u(T)\|^2}{(1+\delta)^2(1+\delta')^2N\|u_{in}\|^2}\geq\frac{1}{81Ng_u^2}\,.
\end{equation} 
and the measured state is $y_1(T)/\|y_1(T)\|$. This is state is close to $u(T)/\|u(T)\|$. To see this, first note that
\begin{equation}
    \|y_1-u(T)\|\leq(\delta+(1+\delta)\delta'\sqrt{N})\|u(T)\|\,.
\end{equation}
Choose $\delta$ and $\delta'$ so that $\delta+(1+\delta)\delta'\sqrt{N}\leq \epsilon/2$. We now have (similar to the proof of \thm{main_thm})
\begin{align}
    &\Big\|\frac{y_1(T)}{\|y_1(T)\|}-\frac{u(T)}{\|u(T)\|}\Big\|\leq\Big\|\frac{y_1(T)}{\|u(T)\|}-\frac{u(T)}{\|u(T)\|}\Big\|+\Big\|\frac{y_1(T)}{\|y_1(T)\|}-\frac{y_1(T)}{\|u(T)\|}\Big\| \\
    &\leq \frac{\epsilon}{2} + \frac{\Big|\|u(T)\|-\|y_1(T)\|\Big|}{\|u(T)\|}\\
    &\leq \epsilon\,.
\end{align}
Now using \thm{main_thm}, we get the query and gate complexities as in the theorem statement.
\end{proof}

\section{Conclusions}\label{sec:conclusions}
We have presented a quantum algorithm to solve linear inhomogeneous ODEs. The gate and query complexity of our algorithm are bounded. Classes of matrices (namely those with non-positive log-norm but with a large condition number for the diagonalizing matrix) where the bounds derived here are exponentially better than in previous work \cite{Ber14, BCOW17} are discussed. Our algorithm is also able to solve the linear ODE for many classes of non-diagonalizable matrices and even singular matrices. We construct a different linear system here because it is easier to apply our analysis to it, but it is interesting to see if the linear system in \cite{BCOW17} has a small condition number (even when $\kappa_V$ is large). Our condition involving the matrix exponential will be particularly useful when considering block matrices such as in nonlinear ODEs. Bounding the matrix exponential of a block matrix in terms of the matrix exponentials of individual matrices is easier than checking for diagonalizability or other spectral conditions.

An interesting open question is whether the dependence on $\|A\|$ can be improved for stable matrices. For example, could $\|A\|$ be replaced by a quantity proportional to the spectral radius $\rho(A)$ of $A$, which can be considerably smaller than $\|A\|$ for some matrices. It is known that $\rho(A)\leq \|A\|$ for all matrices. But for stable matrices, Gelfand's theorem says that for large enough $k$, there exists an $\epsilon$ such that $\|A^k\|\leq (\rho+\epsilon)^k$. To be able to use this to pick the step size $h$, we need to find a quantitative bound relating $\epsilon$ and $k$. The theory of pseudospectra \cite{trefethen2005spectra} gives bounds on norms of powers of matrices. It would be useful to explore this further to see a relationship between the $\epsilon$-pseudosprectral radius $\rho_\epsilon(A)$ and the norm of $A$.


In the realm of nonlinear differential equations, we have exponentially improved the dependence on error of previous work \cite{Liue2026805118}. Our algorithm is also efficient for nonlinear differential equations when $F_1$ is non-normal (and even non-diagonalizable) if it has a negative log-norm. The linear ODE algorithm here can handle even more general cases (i.e., whenever $C(A)$ is bounded and $A$ is sparse). It would be interesting to extend our algorithm for nonlinear differential equations to these more general cases. For this, one needs to extend the definition of $R$ and prove convergence of Carleman linearization to these more general matrices.

\section*{Acknowledgements}
I would like to thank Nuno Loureiro, Andrew Childs, Jin-Peng Liu, Konstantina Trivisa, Paola Capellaro, Abtin Ameri, Herman K\o lden, Erika Ye and Matteo Lostaglio for useful discussions on quantum algorithms and differential equations. I am also grateful to the anonymous referees of CIMP for useful comments that helped improve the paper. An error in Lemma 16 is fixed in this version (independently noticed in \cite{jennings2025quantumalgorithmsgeneralnonlinear}). This material is based upon work supported by the US Department of Energy, Office of Science, Office of Fusion Energy Sciences under award number DE-SC0020264.
\bibliographystyle{ieeetr}
\bibliography{qdiff}

\appendix
\section*{Appendices}
\section[Proof of]{Proof of \protect\lem{Carleman}}\label{appx:appendix_A}
For completeness, in this appendix, we give the proof of the above lemma which is from \cite{Liue2026805118}. We make modifications to the proof to give bounds in terms of the log norm. 

\decrease*
\begin{proof}
Consider an instance of the quadratic ODE, and assume $R < 1$. Let
\begin{equation}
r_{\pm} \coloneqq \frac{-\mu(F_0)\pm\sqrt{\mu(F_0)^2-4\|F_2\|\|F_0\|}}{2\|F_2\|}.
\label{eq:roots}
\end{equation}
We first consider the derivative of $\|u(t)\|$ for which we have 
\begin{align}
    \frac{d\|u\|^2}{dt}&= u^\dag F_2(u\otimes u) + (u^\dag\otimes u^\dag)F_2^\dag u + u^\dag(F_1 + F_1^\dag)u + u^\dag F_0(t) + F_0(t)^\dag u,\nonumber\\
    &\leq  2\|F_2\|\|u\|^3 + 2\mu(F_0)\|u\|^2 + 2\|F_0\|\|u\|.
\end{align}
If $\|u\|\neq 0$, then
\begin{equation}
    \frac{d\|u\|}{dt}\leq \|F_2\|\|u\|^2 + \mu(F_1)\|u\| + \|F_0\|.
\end{equation}
Letting $a=\|F_2\|>0$, $b=\mu(F_1)<0$, and $c=\|F_0\|>0$, we consider a $1$-dimensional quadratic ODE
\begin{equation}
\frac{\d{x}}{\d{t}} = ax^2+bx+c, \qquad
x(0) = \|u_{\mathrm{in}}\|.
\end{equation}
Since $R<1 \Leftrightarrow -b>a\|u_{\mathrm{in}}\|+\frac{c}{\|u_{\mathrm{in}}\|}$, the discriminant satisfies
\begin{equation}
b^2-4ac > \biggl(a\|u_{\mathrm{in}}\|+\frac{c}{\|u_{\mathrm{in}}\|}\biggr)^2-4a\|u_{\mathrm{in}}\|\cdot\frac{c}{\|u_{\mathrm{in}}\|} \ge \biggl(a\|u_{\mathrm{in}}\|-\frac{c}{\|u_{\mathrm{in}}\|}\biggr)^2 \ge 0.
\end{equation}
Thus, $r_\pm$ defined in \eq{roots} are distinct real roots of $ax^2+bx+c$. Since $r_-+r_+ = -\frac{b}{a} > 0$ and $r_-r_+ = \frac{c}{a} \ge 0$, we have $0 \le r_- < r_+$. We can rewrite the ODE as
\begin{equation}
\frac{\d{x}}{\d{t}} = ax^2+bx+c = a(x-r_-)(x-r_+), \qquad
x(0) = \|u_{\mathrm{in}}\|.
\end{equation}
Letting $y=x-r_-$, we obtain an associated homogeneous quadratic ODE
\begin{equation}
\frac{\d{y}}{\d{t}} = -a(r_+-r_-)y+ay^2 = ay[y-(r_+-r_-)], \qquad
y(0) =\|u_{\mathrm{in}}\|-r_-.
\end{equation}
Since the homogeneous equation has the closed-form solution
\begin{equation}
y(t) = \frac{r_+-r_-}{1-e^{a(r_+-r_-)t}[1-(r_+-r_-)/(\|u_{\mathrm{in}}\|-r_-)]},
\end{equation}
the solution of the inhomogeneous equation can be obtained as
\begin{equation}
x(t) = \frac{r_+-r_-}{1-e^{a(r_+-r_-)t}[1-(r_+-r_-)/(\|u_{\mathrm{in}}\|-r_-)]} + r_-.
\end{equation}
Therefore we have
\begin{equation}
\|u(t)\| \le \frac{r_+-r_-}{1-e^{a(r_+-r_-)t}[1-(r_+-r_-)/(\|u_{\mathrm{in}}\|-r_-)]} + r_-.
\label{eq:solution-norm}
\end{equation}
Since $R < 1 \Leftrightarrow a\|u_{\mathrm{in}}\|+\frac{c}{\|u_{\mathrm{in}}\|}<-b \Leftrightarrow a\|u_{\mathrm{in}}\|^2+b\|u_{\mathrm{in}}\|+c<0$, $\|u_{\mathrm{in}}\|$ is located between the two roots $r_-$ and $r_+$, and thus $1-(r_+-r_-)/(\|u_{\mathrm{in}}\|-r_-)<0$. This implies $\|u(t)\|$ in \eq{solution-norm} decreases from $u(0) = \|u_{\mathrm{in}}\|$, so we have $\|u(t)\|<\|u_{\mathrm{in}}\|<r_+$ for any $t>0$.
\end{proof}

\section[Proof of]{Proof of \protect\lem{conditions}}\label{appx:appendix_B}
\conds*
\begin{proof}
As before, for notational simplicity, let $a=\|F_2\|$, $b=\mu(F_1)$ and $c=\|F_0\|$. Note that $b<0$ according to our assumption on the log-norm of $F_1$. Assume that we have an instance with the initial state $u(0)$ such that $\|u(0)\|=\|u_{in}\|$. For this, $R<1$ implies 
\begin{equation}\label{eq:Q_1}
    |b|>\frac{c}{\|u_{in}\|} + a\|u_{in}\|\,.
\end{equation} 
Now consider the following quadratic equation in $x\in \mathbb{R}$.
\begin{equation}
    Q(x)=ax^2 + bx + c\,.
\end{equation}
It has been shown in \lem{Carleman} (see also 3.A.1 in the supplemental information of \cite{Liue2026805118}) that this quadratic equation has non-negative unequal roots, say $r_+$ and $r_-$ i.e., $0\leq r_-<r_+$. In the regime $x\in [r_-,r_+]$, we have $Q(x)<0$. It is also clear that $x=\|u_{in}\|$ lies in between $r_-$ and $r_+$ since 
\begin{equation}
    a\|u_{in}\|^2 + b\|u_{in}\| + c<0\,.
\end{equation} 
The above equation follows from \eq{Q_1}, which is a restatement of the assumption $R<1$. The point $x=\|u_{in}\|$ cannot be one of the end points since $Q(\|u_{in}\|)\neq 0$. This means that there is a $\gamma\in (\|u_{in}\|,r_+)$ such that $Q(\gamma)<0$. Now, if we rescale the quadratic ODE such that the initial state has $\|y(0)\|=\|u_{in}\|/\gamma<1$, the new norms (denoted $\Tilde{a}$, $\tilde{b}$ and $\tilde{c}$) are 
\begin{equation}
    \Tilde{a} = a\gamma\,,\hspace{0.2in} \Tilde{b}=b\,,\hspace{0.2in} \Tilde{c} = \frac{c}{\gamma}\,.
\end{equation}
The fact that $Q(\gamma)<0$ implies that
\begin{equation}
    |b| > \frac{c}{\gamma} + \gamma a\,,
\end{equation}
which means
\begin{equation}
    |\Tilde{b}|>\Tilde{c} + \Tilde{a}\,,
\end{equation}
which satisfies condition (1) using \lem{C(A)_nlode}. As for condition (2), the norm of the initial state satisfies
\begin{equation}
    \|y(0)\|=\frac{\|u_{in}\|}{\gamma}<1\,.
\end{equation}
Finally, we still have $R<1$ since rescaling does not change $R$.
    
\end{proof}

\end{document}